\documentclass[a4paper,11pt]{article}
\pdfoutput=1 %
\usepackage{jheppub} %
\usepackage[T1]{fontenc} %
\usepackage{amsthm}
\usepackage{xcolor}
\usepackage[normalem]{ulem}
\usepackage{comment}
\newtheorem{theorem}{Theorem}[section]

\newtheorem{definition}{Definition}[section]
\newcommand{\dd}{\mathrm{d}}
\newcommand{\deriv}[2]{\dfrac{\dd #1}{\dd #2}}

\title{\boldmath On the Asymptotic Causal Structure in Gravitational EFTs} 

\author[a, b, c]{Bruno Bucciotti,}
\author[d, e, f]{Paolo Creminelli,}
\author[f,g,h]{Alessandro Longo,}
\author[f, g]{Warin Patrick McBlain,}
\author[a, b]{Enrico Trincherini}

\affiliation[a]{Scuola Normale Superiore, Piazza dei Cavalieri 7, 56126, Pisa, Italy}
\affiliation[b]{INFN, Sezione di Pisa, Largo Pontecorvo 3, 56127, Pisa, Italy}
\affiliation[c]{Department of Physics and Beyond: Center for Fundamental Concepts in Science, Arizona State University, 650 E. Tyler Mall,
Tempe, AZ 85281, USA}
\affiliation[d]{ICTP, International Centre for Theoretical Physics, Strada Costiera 11, 34151, Trieste, Italy}
\affiliation[e]{IFPU - Institute for Fundamental Physics of the Universe, Via Beirut 2, 34014, Trieste, Italy}
\affiliation[f]{INFN, Sezione di Trieste, via Valerio 2, 34127 Trieste, Italy}
\affiliation[g]{SISSA, International School for Advanced Studies, via Bonomea, 265, 34136 Trieste, Italy}
\affiliation[h]{Université Paris Cité, CNRS, AstroParticule et Cosmologie, 10 rue Alice Domon et Léonie Duquet, 75013 Paris, France}

\emailAdd{bbucciot@asu.edu}
\emailAdd{creminel@ictp.it}
\emailAdd{alongo@apc.in2p3.fr}
\emailAdd{wmcblain@sissa.it}
\emailAdd{enrico.trincherini@sns.it}

\abstract{It is usually assumed that a healthy EFT should not allow superluminal propagation. In the presence of gravity, however, the notion of superluminality becomes subtle, since there is no invariant way to compare with an \emph{underlying} Minkowski light cone. One can instead resort to an asymptotic criterion: whether the EFT can induce signal propagation faster than what allowed by the asymptotic structure of spacetime.
In this work we study the asymptotic causal structure of gravitational EFTs by analysing signal propagation in black-hole backgrounds in the presence of higher-derivative operators. We show that in spacetime dimensions $D>4$ the effective light cones can lead to genuine asymptotic superluminality, which can be used to constrain the regime of validity of the EFT. By contrast, in $D=4$ the asymptotic causal structure is universally identical to that of Schwarzschild: prompt null curves remain insensitive to higher-derivative corrections and no asymptotic time advance is possible.
We first study the representative operator $R_{\mu\nu\rho\sigma}F^{\mu\nu}F^{\rho\sigma}$, then show that this conclusion is true for any EFT, as it relies only on the asymptotic behaviour of the metric. Finally, we discuss two ways to define superluminality in $D=4$ spacetimes: introducing a covariant cut-off by putting the theory in an asymptotically-AdS background, or imposing a hard cut-off by working at finite distance.
}

\begin{document}

\maketitle
\flushbottom

\section{Introduction}
\label{sec:intro}

A fundamental requirement of any consistent relativistic theory is that signals propagate causally. Starting from the seminal work of Ref.~\cite{Adams:2006sv}, this principle has been shown to have surprisingly powerful consequences, particularly when applied to effective field theories (EFTs).

One approach to probing causality in an EFT is to study the propagation of small fluctuations about a nontrivial classical background. Taking the background to be stationary and focusing on the high-frequency limit, one may employ a geometric optics approximation, in which wave propagation reduces to motion along effective spacetime trajectories determined by the principal part of the equations of motion. Remarkably, theories with Lorentz-invariant and local interactions can admit background configurations that significantly modify the effective causal structure.

This phenomenon is most transparent in flat spacetime in the absence of dynamical gravity. In this case, trajectories exhibiting superluminal propagation relative to the Minkowski light cone admit a sharp interpretation, and one can formulate precise constraints by studying asymptotic time delays.\footnote{See, however, Ref.~\cite{Kaplan:2024qtf} for a different conclusion, in which enhanced quantum effects are argued to render unpredictive precisely those backgrounds that could give rise to causal paradoxes, thereby protecting the causal structure of the EFT.} In particular, requiring that signals do not experience measurable time advances relative to the Minkowski null cone leads to well-defined constraints on the space of EFTs consistent with causality, by bounding the ultraviolet cut-off and/or the allowed range of couplings.

When gravity is dynamical, however, the situation becomes more subtle. A natural first step is to restrict attention to asymptotically flat spacetimes and define asymptotic causality with respect to the causal structure of Minkowski spacetime. In practice, this amounts to requiring that signals sent from past null infinity do not reach future null infinity earlier than they would in flat space.
However, diffeomorphism invariance obscures the notion of a gauge-invariant local time delay, and even asymptotic observables require careful definition. In four spacetime dimensions, as emphasised by Penrose \cite{PENROSE19801} and further elaborated by \cite{Cameron:2020itp}, the long-range nature of gravity and the associated infrared effects make the definition of a sharp asymptotic time delay subtle. By contrast, in higher spacetime dimensions these infrared difficulties are absent, and asymptotic causality can be defined in an unambiguous manner.

To make this discussion more concrete, a sharp question in this context is the following: given an isolated black hole spacetime, can signals propagate faster than allowed by the asymptotic causal structure of the metric? In General Relativity, the answer is negative. The Gao–Wald theorem \cite{Gao:2000ga} implies that, assuming the null energy condition and a suitable global condition, any causal curve connecting past and future null infinity cannot arrive earlier than the corresponding null geodesic in the asymptotic region.

Once higher-derivative operators are included, however, the conclusion of the Gao–Wald theorem need no longer apply. In particular, effective field theory corrections generically modify the equations of motion so that the propagation of high-frequency fluctuations is governed by an effective metric that differs from the background geometry. From this perspective, causality is controlled not solely by the spacetime metric, but by the characteristic structure of the EFT~\cite{Babichev:2007dw}. It then becomes a meaningful and non-trivial question whether such corrections can generate measurable asymptotic time advances, allowing signals to reach future null infinity earlier than in the Schwarzschild background.

This question was the original motivation behind the analysis of Ref.~\cite{Camanho:2014apa}, which investigated whether the propagation of a photon or a graviton in a shock-wave background can exhibit an asymptotic time advance in the presence of higher-derivative interactions. Since then, a number of works \cite{Cheung:2014ega, Papallo:2015rna, Hollowood:2015elj, Benakli:2015qlh, Goon:2016une, deRham:2020zyh, AccettulliHuber:2020oou, Chen:2021bvg, deRham:2021bll, Bellazzini:2021shn, Serra:2022pzl, Bellazzini:2022wzv,  Chen:2023rar, Cremonini:2023epg, Cremonini:2024lxn, Grojean:2025zbn, Deb:2026enc} have employed related ideas and similar backgrounds to study time delays in gravitational EFTs.

In this work, we reconsider this question by analysing the asymptotic causal structure directly in terms of prompt null curves in the presence of higher-derivative operators, focusing in particular on the difference between $D=4$ and $D>4$ spacetime dimensions. To be concrete, in Section~\ref{FFR} we consider the EFT of a photon coupled to gravity through the operator $\alpha R_{\mu\nu\rho\sigma}F^{\mu\nu}F^{\rho\sigma}$. It has long been known \cite{Drummond:1979pp} that this interaction modifies the photon effective metric in such a way that propagation outside the background metric null cone can occur, depending on the polarisation and direction of the photon in a Schwarzschild background. As an irrelevant operator, its effects become important at short distances, namely below the length scale $\sqrt{\alpha}$.

For sufficiently small black holes, with Schwarzschild radius $r_s < \sqrt{\alpha}$, we find that in $D>4$ there exist light ray trajectories connecting points on past and future null infinity and passing closer than $\sqrt{\alpha}$ from the black hole that arrive earlier than the corresponding Schwarzschild-like null geodesic connecting the same endpoints. This constitutes a genuine violation of asymptotic causality. Such violations can only be avoided if the UV cut-off of the theory satisfies $\Lambda < 1/\sqrt{\alpha}$, parametrically below the perturbative strong coupling scale $(M_P/\alpha)^{1/3}$. This ensures the additional UV corrections become important before entering the regime where time advances would arise.
In contrast, in $D=4$ we show that the prompt curve is always the Schwarzschild-like null geodesic. Therefore, the asymptotic causal structure in $D=4$ coincides with that of General Relativity, independently of the presence of the $FFR$ operator.

In Section~\ref{theorems}, we demonstrate that this result is universal and does not depend on the specific form of the higher-derivative interaction. In four spacetime dimensions, for {\it any} asymptotically Schwarzschild spacetime, prompt null curves never probe the near-horizon region. Consequently, the asymptotic causal structure is insensitive to irrelevant operators localised at short distances. The physical origin of this behaviour can be traced to the logarithmically IR divergent contribution to the time of flight.

This divergence motivates the question of whether introducing a natural IR regulator can restore meaningful causality bounds in $D=4$.  A covariant way to implement an IR regulator is to consider asymptotically-AdS backgrounds, where the AdS radius $L$ provides a characteristic scale. In this case, Gao and Wald established a theorem~\cite{Gao:2000ga} showing that prompt causal curves between two points on the conformal boundary of asymptotically-AdS spacetimes satisfying the null energy condition must lie entirely on the boundary in GR. In Section~\ref{AdS}, we study four-dimensional AdS-Schwarzschild backgrounds and show that the AdS regulator does render individual calculations finite and could in principle allow for causality constraints. However, in the large $L$ limit, the IR effects that obstruct a non-trivial notion of asymptotic causality in flat space reappear: null geodesics remain far from the black hole, and gravitational time delays again dominate any finite time advances from beyond-GR corrections. Thus, causality bounds derived in AdS remain fundamentally dependent on the IR regulator and cannot be smoothly connected to flat-space physics. In Section \ref{finitedistance} we come back to the asymptotically flat case and consider the propagation of signals over a finite distance. If this distance is large enough that the logarithm in the time-delay expression is much greater than unity, coordinate ambiguities are avoided. We argue that superluminality in this setup is not related to closed timelike curves in any obvious way: building a time machine is obstructed by the requirement for parametrically large boosts.
\section{\texorpdfstring{$FFR$}{FFR} Theory}
\label{FFR}
We begin by studying the low energy theory of light and gravity coupled through an irrelevant $FFR$ operator, where $F$ denotes the the electromagnetic field strength and $R$ the Riemann tensor. This is the only operator contributing at leading order on Ricci-flat backgrounds, up to field redefinitions, assuming parity. The EFT action is therefore
\begin{equation}
    S = \int \dd^D x \sqrt{-g} \bigg[\frac{M_{P}^2}{2}R - \frac{1}{4}F^{\mu\nu}F_{\mu\nu} + \alpha R_{\mu\nu\rho\sigma}F^{\mu\nu}F^{\rho\sigma}\bigg]\,. \label{eq:L_FFR}
\end{equation}
It is readily observed that, in the absence of electromagnetic fields, the Schwarzschild metric 
\begin{equation}
    \dd s^2 = - \left(1 - \frac{r_s^{D-3}}{r^{D-3}}\right) \dd t^2 + \left(1 - \frac{r_s^{D-3}}{r^{D-3}}\right)^{-1} \dd r^2 + r^2 \dd \Omega_{D-2}^2 \label{eq:schwarz_D}
\end{equation}
remains a solution.

Importantly, the action \eqref{eq:L_FFR} gives second-order equations of motion (the well-posedness of the corresponding initial-value problem is studied in \cite{Davies:2021frz}). There are no extra ghost-like degrees of freedom with mass of order $\alpha^{-1/2}$, which would clearly indicate the breaking of the EFT at that scale.\footnote{The propagation of perturbations in a background geometry at linear order in the Riemann tensor is always described by second-order equations of motion. Locally, one can look at the Fourier modes of the perturbation with momentum $k^\mu$. A contribution to the equations of motion with more than two derivatives would give more than two momenta $k^\mu$. These vectors cannot be contracted with polarisation vectors because of transversality, and they cannot be contracted with themselves because $k^2=0$ and one can remove this using a field redefinition. Therefore all the indices should be contracted with the Riemann tensor, but this is not possible taking into account its symmetries. The argument applies to perturbations of any spin. This is also clear from the scattering amplitude perspective of \cite{Camanho:2014apa}.} 

The equation of motion for the photon
\begin{equation}
    \nabla_\nu F^{\mu\nu} - 4\alpha R^{\mu\nu}_{~~~\rho\sigma}\nabla_{\nu}F^{\rho\sigma} = 0 \label{eq:eom_FFR}
\end{equation}
shows that the $FFR$ term induces a polarisation-dependent force. Consequently, the EFT \eqref{eq:L_FFR} can exhibit propagation of signals faster than a massless field which is minimally coupled to the metric \cite{Drummond:1979pp,Goon:2016une}. 
In contrast to EFTs where superluminality can be eliminated by imposing sign constraints on the coupling coefficients \cite{Adams:2006sv}, the theory \eqref{eq:L_FFR} admits, for either sign of $\alpha$, a choice of polarisation that leads to superluminal photon propagation. The only way to avoid such behaviour is therefore to require that, for a given value of $\alpha$, the cut-off of the EFT is sufficiently low such that any potential superluminal effects become unresolvable.
Indeed, this is the resolution of the so-called Drummond–Hathrell `paradox', in which \eqref{eq:L_FFR} appears as the low-energy EFT of QED in curved spacetime after integrating out the electron. Moreover, any attempt to parametrically enhance the time-advance effect within this EFT inevitably leads to a loss of EFT control \cite{Goon:2016une}. This indicates that such enhancements are forbidden within the EFT's regime of validity, which is consistent with the expectation that the UV theory is free of pathologies.

In the remainder of this section, we determine how photons propagate in black-hole backgrounds and analyse their possible trajectories.
\subsection{Effective Metrics}
Due to their non-minimal coupling to gravity, photons in $FFR$ theory do not propagate along the null geodesics of the Schwarzschild metric. To investigate the resulting causal structure, we study the propagation of signals, specifically high-frequency (wave-packet) photon excitations.
In general, the causal structure of a theory is governed by the principal part of the linearised equations of motion around a background (the terms carrying the highest number of derivatives on the perturbation) since these determine the propagation of sharp wavefronts. For the $FFR$ operator, the principal part contains two derivatives acting on the perturbation, as the kinetic term. More generally, the action (\ref{eq:L_FFR}) should be understood as the leading part of an EFT, which will contain additional higher-dimensional operators that can, in principle, contribute additional derivatives to the principal part. However, provided the typical inverse wavelength of the wave packet is large compared to the inverse curvature radius of the background, while still remaining below the scale suppressing these operators, their contributions are subleading and can be consistently neglected.

The resulting leading-order equation, explicitly given in \eqref{characteristiceq}, is known in the theory of partial differential equations as the \emph{characteristic equation}. It governs the propagation of high-frequency disturbances and, in the present context, captures also the causal structure of the theory in the WKB or geometric-optics limit. As discussed in Appendix \ref{app:FFR_effectivemetric}, in $FFR$ theory the characteristic equation for the physical photon modes factorises into a product of effective metrics, one for each polarisation. The `causal cone' of the theory is defined by the outermost (i.e.~widest) null cone among the effective metrics.
\begin{figure}[h]
    \centering
    \includegraphics[width=0.4\linewidth]{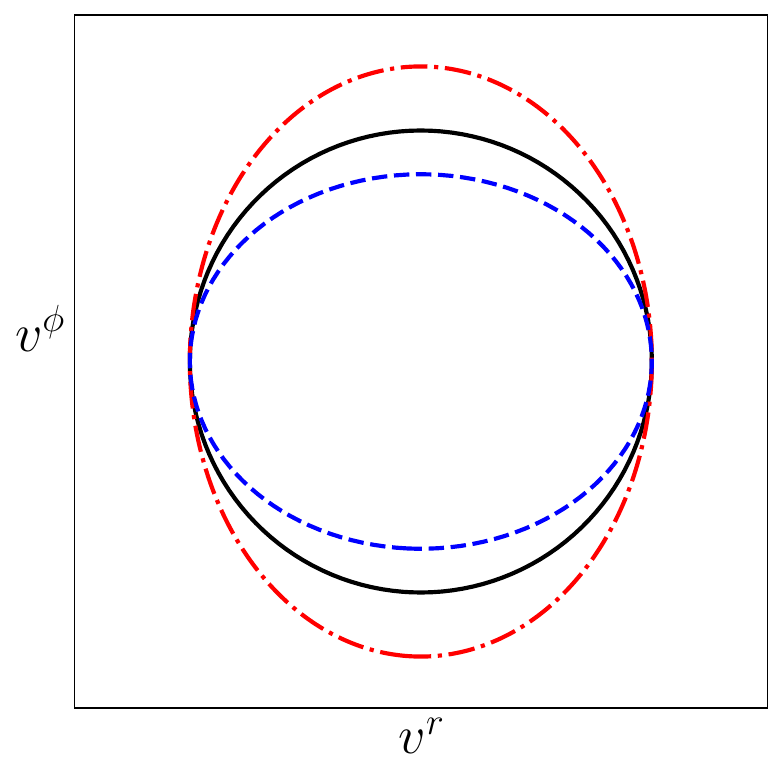}
    \caption{A cross-section of the nested set of light cones described by \eqref{eq:eff_metric} for black holes in $D = 4$, with axes scaled such that the background metric (shown in black) is a circle. The null cones for the larger (dashed dot red) and smaller (dashed blue) effective metrics corresponding to two different polarisation modes are shown as well. Figure adapted from \cite{Reall:2014pwa}.}
    \label{fig:nested_cones}
\end{figure}
Going through the analysis, the full details of which can be found in Appendix \ref{app:FFR_effectivemetric}, we therefore obtain the \emph{scalar} and \emph{vector} effective metrics of the photon governed by \eqref{eq:eom_FFR}:
\begin{equation} \label{eq:eff_metric}
    \mathrm{d}s^2 = -\left(1 - \frac{r_s^{D-3}}{r^{D-3}}\right) \mathrm{d}t^2 + \left(1 - \frac{r_s^{D-3}}{r^{D-3}}\right)^{-1} \mathrm{d}r^2 + \frac{r^2}{c_{S/V}(r)} \mathrm{d} \Omega^2_{D-2} 
\end{equation}
with
\begin{equation}
    c_S(r) = \frac{1 + \frac{c_1 \alpha r_s^{D-3}}{r^{D-1}}}{1 - \frac{c_2 \alpha r_s^{D-3}}{r^{D-1}}} \qquad \text{and} \qquad
    c_V(r)=\frac{1 - \frac{8 \alpha r_s^{D-3}}{r^{D-1}}}{1 + \frac{c_1 \alpha r_s^{D-3}}{r^{D-1}}} \>,\label{eq:eff_metric_C}
\end{equation}
where the coefficients $c_1$ and $c_2$ depend on the number of spacetime dimensions. Their values are
\begin{equation}
    c_1 = 4(D-3)\qquad
    c_2 = 4(D-3)(D-2)\>.
\end{equation}
To cast the metric in the form \eqref{eq:eff_metric}, a Weyl rescaling has been used to conveniently isolate the effect of the FFR operator entirely within the angular part of the metric\footnote{Notice that we did not perform any change of coordinates: we are still in the original Schwarzschild coordinates.}. In contrast, a minimally coupled field would propagate according to the standard Schwarzschild metric, which simply corresponds to setting $c_{S/V}=1$. It follows that the null cones of the metric characterised by $c_{S/V}>1$ enclose the null cones of the standard Schwarzschild metric: see Figure \ref{fig:nested_cones}.

The two metrics above govern the propagation of states that transform like scalars (longitudinal polarisation) and vectors (transverse polarisation) under $SO(D-1)$, respectively. For a given sign of $\alpha$, there is a clear hierarchy between the cones of the effective metrics. When $\alpha>0$, we have $c_S(r)>1>c_V(r)$, while for $\alpha<0$ the hierarchy is reversed $c_V(r)>1>c_S(r)$. Thus, depending on the sign of $\alpha$, one polarisation propagates faster than the other. For $\alpha>0$, the null cone of the \emph{scalar} effective metric is broader than the underlying Schwarzschild, while the transverse polarisations propagate strictly inside the null cone of the Schwarzschild metric (the situation is reversed when $\alpha<0$). Notice that photons of any polarisation will propagate at the speed of light radially. The widest light cone of the effective metric \eqref{eq:eff_metric} defines causality in this theory, and we will refer to it as the \emph{effective light cone}. That the effective light cone locally extends beyond the Schwarzschild null cone is not, by itself, a pathology. In a theory with dynamical gravity, there is no fixed background light cone against which to measure superluminality: a diffeomorphism can always set any chosen metric to Minkowski form at a given point, so the local comparison between two metrics is gauge-dependent. A physically meaningful question, that we will explore in the rest of the paper, is instead whether signals can arrive earlier than permitted by the asymptotic causal structure of the spacetime. 
It is worth mentioning that, at linear order in $\alpha$, two expressions in \eqref{eq:eff_metric_C} reduce to
\begin{equation}
    c_S(r) \approx 1+ (c_1+c_2)\frac{\alpha r_s^{D-3}}{r^{D-1}} \qquad \text{and} \qquad
    c_V(r) \approx 1- (c_1+8)\frac{\alpha r_s^{D-3}}{r^{D-1}} \>,
\end{equation}
As a consequence, the sign of $\alpha$ is effectively irrelevant at this order, and therefore may be chosen to be positive without loss of generality. For this reason, we will fix $\alpha>0$ for the rest of the discussion, which means that the broadest effective metric is the `scalar' one ($c_S>1$). All the results can be mapped to the case $\alpha<0$ simply by replacing $c_1$ with $-8$ and $c_2$ with $-c_1$. 
\subsection{Phenomenology}\label{phenomenology}
The EFT in \eqref{eq:L_FFR} introduces two additional scales in the system. The first is the length scale associated with the 
$FFR$ coupling, $\sqrt{\alpha}$. The second is the energy scale $(M_P/\alpha)^{1/3}$ $(\gg 1/\sqrt{\alpha})$ at which the new interaction becomes strongly coupled and which, therefore, sets an upper bound on the UV cut-off of the theory. As stated in the introduction, our primary goal is to determine whether causality constraints can be used to infer that the cut-off scale must lie parametrically below this perturbative estimate. 

Since we are interested in the causal structure experienced by photons propagating in black hole spacetimes, it is useful to distinguish between two regimes. A black hole can be considered \emph{large} if $r_s \gtrsim \sqrt{\alpha}$ or \emph{small} if $r_s \lesssim \sqrt{\alpha}$. 

To better understand the impact of the $FFR$ term, in Appendix \ref{app:general_metric} we study the effective potential experienced by a photon moving along null geodesics with angular momentum $\ell$. Using equations \eqref{eq:v_eff} and \eqref{eq:eff_metric_C}, we find
\begin{equation}
    V_{\text{eff}}(r) = \frac{\ell^2}{2r^2}\left(1-\frac{r_s^{D-3}}{r^{D-3}}\right)\frac{1 + \frac{c_1 \alpha r_s^{D-3}}{r^{D-1}}}{1 - \frac{c_2 \alpha r_s^{D-3}}{r^{D-1}}}. \label{eq:v_eff_FFR}
\end{equation}
\begin{figure}
    \centering
    \includegraphics[width=0.5\linewidth]{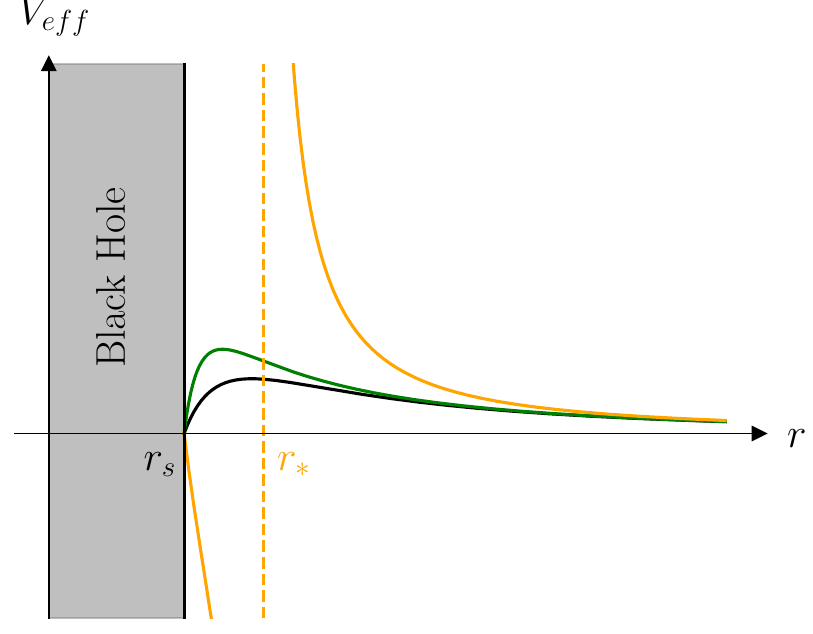}
    \caption{The effective potentials shown for $D=4$. Schwarzschild in black, large black holes in green, small black holes in orange. The potential barrier for the small black hole case is located at the vertical line $r \sim r_*$.}
    \label{fig:v_eff_FFR}
\end{figure}
The shape of the potential is given in Figure \ref{fig:v_eff_FFR}. From this, we see that for large black holes ($r_s > \sqrt{\alpha}$), the potential still retains a Schwarzschild-like shape, implying that photons can scatter off or enter into the black hole. The inclusion of the $FFR$ term shifts the peak of the potential to the left, reducing the radius of the photon sphere from the usual Schwarzschild value.

For small black holes, by contrast, the potential develops an infinite barrier at the length scale 
\begin{equation}
    r_* = (\alpha r_s^{D-3})^{\frac{1}{D-1}}, \label{eq:r_*}
\end{equation}
indicating that the $FFR$ effect is repulsive. To identify the regime in which this contribution dominates, we expand \eqref{eq:v_eff_FFR} for $r \gg r_*$
\begin{equation}
    \frac{V_{\text{eff}}(r)}{\ell^2} \approx \frac{1}{2r^2} - \frac{r_s^{D-3}}{2r^{D-1}} + \frac{1}{2}(c_1 + c_2) \frac{\alpha r_s^{D-3}}{r^{D+1}}.
\end{equation}
The first term corresponds to the centrifugal barrier, the second to the Schwarzschild contribution, and the third arises from the 
$FFR$ operator. For the $FFR$ term to dominate over the Schwarzschild term, we must have $r \lesssim \sqrt{\alpha}$, implying that the black hole must be small for this effect to be significant.

The scale $r_*$ possesses an additional crucial property: it marks a breakdown of causality in a local sense, distinct from the asymptotic notion studied in this paper. At $r = r_*$, the function $c(r)$ in the effective metric \eqref{eq:eff_metric} diverges, signalling a breakdown of hyperbolicity; below this radius, $c(r)$ turns negative and the effective metric changes signature, so the notion of causality defined by the effective light cone ceases to be meaningful altogether. This implies that additional higher-derivative operators must enter at the scale $r_*$, contributing corrections to the propagation speed comparable to those of the $FFR$ term — the EFT breaks down there, constraining its cut-off. Applied to the $FFR$ theory, this gives a cut-off bound $\Lambda \lesssim 1/r_* \ll (M_P/\alpha)^{1/3}$ from a purely local argument, independent of asymptotic causality, so the strong-coupling scale is already ruled out as the UV cut-off of the theory.

In general, however, the scale at which hyperbolicity breaks down and the scale associated with the leading higher-dimensional operator affecting the propagation speed are distinct. In the rest of the paper, we therefore focus on whether stronger constraints on the EFT cut-off can be obtained from asymptotic causality by studying photons that propagate sufficiently close to the black hole to be sensitive to the additional operator, while remaining far enough from the region where hyperbolicity breaks down.

\subsection{Null Geodesics} \label{ssec:null_geodesics}
Having understood the effects of the $FFR$ operator and the propagation of null rays in this theory---which together determine its local causal structure---we can now turn to its global causal structure. The latter is encoded in prompt causal curves, defined as those curves that connect two spatial points in the least amount of coordinate time. (This definition is valid for stationary spacetimes, which is the one we are limiting to in this paper: for a more general, geometric definition, see e.g.~\cite{Witten:2019qhl}.) For a Lorentzian manifold these curves are null geodesics. However this does not hold when there exist regions where the effective metric changes signature, as discussed above. 

This motivates us to consider the following setup as shown in Figure \ref{fig:FFR-geodesic}: consider two spatial points $p(R, \Tilde{\phi})$ and $q(R, \pi - \Tilde{\phi})$ lying on a straight line defined by $R\sin\Tilde{\phi} = d$ on the equatorial plane. We would like to find all possible null geodesics that connect $p$ and $q$. Before we proceed, we must explain certain quantities that will be used throughout this subsection. The \emph{bending angle} $\Delta\phi$ is defined as $\Delta\phi = \pi - 2\Tilde{\phi}$. Null geodesics will start from point $p$ at radius $R$, approach the black hole at a \emph{distance of closest approach} $r_0$, before reaching point $q$.  Finally, given the spherically symmetric nature of our problem, all of the geodesics will lie on a plane, which we take to be the equatorial plane without loss of generality. 

\begin{figure}[h]
    \centering
    \includegraphics[width=\linewidth]{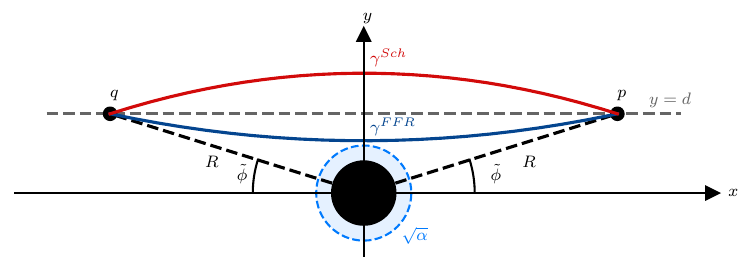}
    \caption{The two possible null geodesics connecting $p$ and $q$ as both points are sent to infinity along the horizontal line $y = d$ for small black holes.}
    \label{fig:FFR-geodesic}
\end{figure}

The explicit calculations are deferred to Appendix \ref{app:perturb_calcs_FFR} and performed at first order in $\frac{r_s^{D-3}}{r_0^{D-3}}$ and $\frac{\alpha r_s^{D-3}}{r_0^{D-1}}$. Applying \eqref{eq:app_deltaphiT_integrals} for $R \gg r_0 \gg r_s$, the bending angle of the null geodesic from $p$ to $q$ is given by \eqref{eq:bending_angle_full},
\begin{equation}
    \Delta \phi \approx \pi - 2\sin^{-1} \left(\frac{r_0}{R}\right) + \left[\left(\frac{r_s}{r_0}\right)^{D-3} - (c_1 + c_2)\frac{D-2}{D-1}\frac{\alpha r_s^{D-3}}{r_0^{D-1}}\right] \frac{\Gamma(\frac{D}{2})}{\Gamma(\frac{D+1}{2})} \frac{\sqrt{\pi} (D - 1)}{2}.
\end{equation}
Using the fact that $\Delta \phi = \pi - 2\Tilde{\phi}$,
\begin{equation} \label{eq:bending_angle}
    \Tilde{\phi} = \sin^{-1} \left(\frac{r_0}{R}\right) - \left[\left(\frac{r_s}{r_0}\right)^{D-3} - (c_1 + c_2)\frac{D-2}{D-1}\frac{\alpha r_s^{D-3}}{r_0^{D-1}}\right] \frac{\Gamma(\frac{D}{2})}{\Gamma(\frac{D+1}{2})} \frac{\sqrt{\pi}(D - 1)}{4}.
\end{equation}
In order to probe the asymptotic causal structure of the theory, we have to scan through all the possible null geodesics connecting $p$ and $q$ in the limit in which the two points are sent to infinity. To this end, we consider the limit in which the two points are moved away along a straight line, that is taking $R \rightarrow \infty$ and $\Tilde{\phi}\rightarrow 0$ while keeping $R\sin\Tilde{\phi} = d$ constant.

For a small black hole we can identify two curves\footnote{More precisely, for two points $p$ and $q$ lying on the same half of the equatorial plane, there exist infinitely many null curves connecting them. These trajectories differ by the number of times they wind around the black hole and by the orientation of their motion. Since our goal is to identify the prompt null curves between the two points, we restrict attention to non-spiralling null curves that remain within the half-plane containing $p$ and $q$.}. The Schwarzschild-like geodesic, which exists also for $\alpha =0$, has a distance of closest approach that grows indefinitely with $R$, i.e. 
\begin{equation}
    r_0^{Sch} \approx \left[\frac{\sqrt{\pi} (D - 1)}{4} \frac{\Gamma(\frac{D}{2})}{\Gamma(\frac{D+1}{2})} R r_s^{D - 3}\right]^{\frac{1}{D-2}}\,. \label{eq:r0_Sch_Line}
\end{equation}
The second geodesic, exclusive to the $FFR$ theory, maintains its proximity to the $\sqrt{\alpha}$ region,
\begin{equation}
    r_0^{FFR} \approx \sqrt{(c_1 + c_2) \alpha \frac{D-2}{D-1}}\,. \label{eq:r0_FFR_Line}
\end{equation}

\subsection{Time Delay/Advance} \label{ssec:tof}
Once all null geodesics connecting the two points have been identified, we compare their respective times of flight in the limit where $p$ and $q$ approach infinity along a straight line, keeping $d=R\sin{\tilde{\phi}}$ fixed.
\subsubsection{\texorpdfstring{$D>4$}{D>4}}
The geodesic time of flight is calculated perturbatively, meaning for $R \gg r_0 \gg r_s$, in \eqref{eq:D_t_geo} up to $\mathcal{O}(r_s)$. Using \eqref{eq:r0_Sch_Line} and \eqref{eq:r0_FFR_Line}, and working in the limit where $p$ and $q$ are taken to infinity, the difference in time of flight between the two geodesics is
\begin{equation}
     \lim_{R\rightarrow\infty}(\Delta t_{geo}^{Sch} - \Delta t_{geo}^{FFR}) = - \frac{r_s^{D-3}}{(r_0^{FFR})^{D-4}} \sqrt{\pi} \frac{D-1}{(D-2)(D-4)} \frac{\Gamma\left(\frac{D}{2}\right)}{\Gamma\left(\frac{D+1}{2}\right)}\,,
\end{equation}
which shows that the Schwarzschild-like geodesic is faster compared to the $FFR$ geodesic.
The reader can notice that the above expression does not depend explicitly on $d$ at leading order in the large $R$ expansion: indeed we might have equivalently chosen $p$ and $q$ to be strictly antipodal, meaning $\tilde\phi\equiv 0$ identically in $R$, without changing the outcome at this order.

At this point, the reader might wonder whether the fact that the Schwarzschild-like geodesic is faster than the $FFR$ geodesic implies the absence of any causality issues, since the first curve probes a region far from the black hole and is therefore insensitive to higher-dimension operators like $FFR$. This conclusion, however, would be premature: while the fastest geodesic is indeed Schwarzschild-like, the prompt causal
curve need not be a geodesic, and faster non-geodesic null trajectories do exist. To see this, we consider a simple non-geodesic curve: a straight horizontal line on the equatorial plane satisfying $R\sin\phi = d$ and $\dd s^2 = 0$. If we consider the difference in time of flight between the Schwarzschild-like geodesic and this null straight line, whose $\Delta t_{str}$ is calculated perturbatively in \eqref{eq:D_t_str}, we get
\begin{equation}
    \lim_{R\rightarrow\infty}(\Delta t_{geo}^{Sch} - \Delta t_{str}) = - d \left(\frac{r_s}{d}\right)^{D-3} \sqrt{\pi} \frac{D-1}{D-4} \left[1 - \frac{\alpha}{d^2} (c_1 + c_2)\frac{D-4}{D-1}\right] \frac{\Gamma\left(\frac{D}{2}\right)}{2\Gamma\left(\frac{D+1}{2}\right)}.
\end{equation}
We see that if this straight line passes through the $FFR$ dominated region,
\begin{equation}\label{straightdist}
    d \lesssim \sqrt{(c_1 + c_2) \alpha \frac{D-4}{D-1}} 
\end{equation}
the null straight path is faster than the two geodesics connecting $p$ and $q$. Notice that the geodesic that stays close to the black hole remains at a distance (see \eqref{eq:r0_FFR_Line}) which is larger than the limit just discussed in \eqref{straightdist}. The geodesic does not get close enough to `profit' off the enhanced speed and win against the far-away geodesic. 

At this point the reader might be confused on why there can be curves which are faster than geodesics. For stationary spacetimes, see for instance \cite{Landau:1975pou}, it can be shown that null curves connecting the two spatial points are geodesics if the path extremises the time-of-flight functional
\begin{equation}
    \delta \int_{p}^{q} \dd t = 0
\end{equation}
which is the same as Fermat's principle. Then, assuming that this functional has a minimum, we claim that a null geodesic is the prompt causal curve connecting the two points. This argument fails in our case because, if we prevent curves from entering the hyperbolicity breaking region $r_*$ where the functional is ill-defined, the space of curves develops a boundary. This boundary consists of curves that come into contact with $r_*$, and the curve minimizing the time-of-flight can belong to this boundary without being a stationary point. As a result, the geodesics found from extremising the time-of-flight functional are merely stationary paths and do not correspond to prompt causal curves.
\subsubsection{\texorpdfstring{$D=4$}{D=4}}\label{sec:D4}
We now repeat the same analysis for the case $D=4$. The geodesic time of flight is calculated perturbatively in \eqref{eq:4d_t_geo} up to $\mathcal{O}(r_s)$. Using \eqref{eq:r0_Sch_Line} and \eqref{eq:r0_FFR_Line}, and working in the limit where $p$ and $q$ are taken to infinity, the difference in time of flight between the two geodesics is
\begin{equation}
     \lim_{R\rightarrow\infty}(\Delta t_{geo}^{Sch} - \Delta t_{geo}^{FFR}) =  \lim_{R\rightarrow\infty}\left[r_s \ln \left(\frac{8\alpha}{R r_s}\right) + \frac{3}{2} r_s - \frac{24 \alpha}{R}\right] = -\infty
\end{equation}
which shows that the Schwarzschild-like geodesic is faster compared to the $FFR$ geodesic. However this conclusion held true even in $D>4$, and the straight line was necessary to find an even faster curve. The difference in time of flight between the Schwarzschild-like geodesic and the null straight line, the latter of which is calculated for $D = 4$ in \eqref{eq:4d_t_str}, is 
\begin{equation}\label{eq:Schvsstraight}
     \lim_{R\rightarrow\infty}(\Delta t_{geo}^{Sch} - \Delta t_{str}) =  \lim_{R\rightarrow\infty}\left[2r_s \ln \left(\frac{d}{\sqrt{R r_s}}\right) + 2 r_s + \frac{8\alpha r_s}{d^2} \left(1 - \frac{d^2}{R r_s}\right)\right] = -\infty
\end{equation}
which shows that the Schwarzschild-like geodesic is still faster than these straight lines. The conclusion is that, among the curves we checked, the Schwarzschild-like geodesic is the fastest. 

In fact, our set of curves was somewhat restricted: we considered only geodesics and straight-line trajectories with impact parameters that admit a perturbative treatment. In the next section, we will give a general argument showing that for $D = 4$, the Schwarzschild-like geodesic is the fastest curve connecting the two points i.e.~the prompt causal curve.

As a final remark, one may choose to take the $R\to \infty$ limit in a different way, namely, sending the two points to infinity along a radial line ($R \rightarrow \infty$) while keeping the initial angle constant ($\Tilde{\phi} = \text{const}$). In this limit, it can be shown that for small black holes, the only geodesics connecting them are the Schwarzschild-like geodesic with $r_0^{Sch} \approx R \sin \Tilde{\phi}$ and another $FFR$ geodesic passing close to the hyperbolicity-breaking region $r_0^{FFR} \sim r_*$. It is tempting to expect that the latter geodesic might be the fastest null curve, because the distance of closest approach lies well inside the $\sqrt{\alpha}$ region where time advances are expected to occur. 
However the former (approximately straight) geodesic that always remains far from the black hole is always faster.
Its time of flight—dominated by the Euclidean distance between the initial and final points—is much less than that of the $FFR$ geodesic due to the Euclidean triangle inequality. Thus, we have
\begin{equation}
    \Delta t_{geo}^{Sch} < \Delta t_{geo}^{FFR}
\end{equation}
implying that the Schwarzschild-like geodesic is faster than the $FFR$ geodesic even in this case. As for the Schwarzschild-like geodesic and the null straight curve, given that the null straight curve is well outside the region where beyond GR effects are important $r\sim R\gg \sqrt{\alpha}$, the usual theorems guaranteeing that a geodesic is the fastest curve apply and the Schwarzschild-like geodesic is the faster curve.
\smallskip

We now compare our results with existing approaches to causality bounds. These approaches typically compute the time delay (defined as the difference between the time of flight in the full theory and the corresponding time in Minkowski spacetime, identified via a specific choice of coordinates) and declare a pathology if this delay is negative and resolvable within the EFT regime. In practice, these calculations evaluate the time of flight along a straight-line trajectory at a fixed impact parameter $b$ (the Born approximation), which is valid in the weak-field regime $b \gg r_s$ where the bending angle is $\mathcal{O}(r_s/b)$. Within this framework, the eikonal phase is computed at a fixed impact parameter, and the time delay (or advance) is then extracted from the derivative of the eikonal phase with respect to energy.

In $D>4$, this procedure is physically well-motivated. The integral defining the Shapiro delay converges rapidly as the endpoints recede from the black hole, since corrections to Minkowski fall off as $1/r^{D-3}$. One can therefore extend a geodesic arc, computed near the black hole at finite impact parameter and scattering angle $\phi$, out to large $R$. This extension contributes a negligible additional time delay, and the resulting causal curve is faster than the asymptotic Minkowski null ray, constituting a genuine violation of asymptotic causality.

In $D=4$ the situation is more subtle. The straight-line time of flight at fixed $b$ receives a gravitational delay $\sim r_s \ln(R/b)$ that grows logarithmically with the distance $R$ to the endpoints, as contributions accumulate at all distances from the source. As $R \to \infty$ the straight-line curve, like any curve that probes the near-horizon region, is never asymptotically prompt. As we have shown, for asymptotically distant endpoints, the Schwarzschild-like geodesic, which stays far from the black hole and avoids accumulating a large delay, is always faster. The Born approximation masks this because it computes the time of flight along a prescribed path without verifying whether that path is actually the fastest one available. 

At finite but large $R$, the logarithm is finite, the Born approximation remains valid, and a net time advance relative to flat space can be obtained, as discussed in the literature. However, this is not a violation of {\it asymptotic} causality and, as we will argue in Section \ref{finitedistance}, the usual pathological implications are less clear in this regime.

\section{Causality Theorems}
\label{theorems}
Previously, in $FFR$ theory, we have seen that for $D>4$ we can find null curves which are faster than all possible null geodesics connecting the two faraway points. On the contrary, the Schwarzschild-like geodesic is truly the fastest null curve in $D = 4$. Prudent readers might be sceptical of this conclusion. Given that we worked perturbatively, one might argue that there might be non-perturbative curves, i.e.~curves that are very close to $r_*$ in $D = 4$, that are not captured by our calculations.

We are going to show that what we found above is actually very general. In $D=4$ for a generic EFT, the Schwarzschild-like geodesic is the prompt causal curve: the asymptotic causal structure is not altered by beyond-GR effects.

In this section, we review a theorem proposed by Gao and Wald \cite{Gao:2000ga} that establishes the asymptotic causal properties for General Relativity. Then, we propose an extension of Gao and Wald to gravitational EFTs in four-dimensional stationary spacetimes.

\subsection{General Relativity: Gao-Wald Theorem}
\begin{figure}[h]
    \centering
    \includegraphics[width=\linewidth]{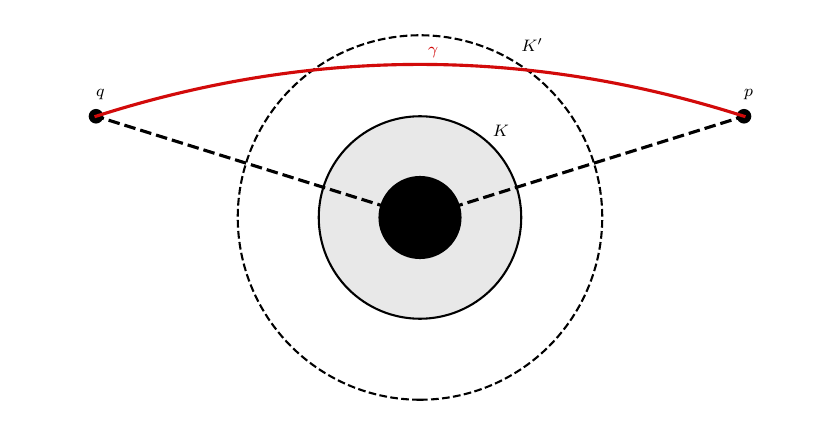}
    \caption{Illustration of the Gao-Wald theorem for Schwarzschild spacetime on the equatorial plane.}
    \label{fig:gao-wald-theorem}
\end{figure}

In General Relativity, Gao and Wald established a theorem in \cite{Gao:2000ga} which states that for a spacetime obeying the null energy condition and a suitable global condition, no time advance is possible. For stationary spacetimes, the statement goes: given a compact region $K$, one can always choose another compact region $K'$ containing $K$ such that, if two points $p$ and $q$ are not in $K'$, then the prompt causal curve connecting $p$ and $q$ does not enter $K$\footnote{In the original formulation of Gao and Wald, the regions $K$ and $K'$ are compact in spacetime. For static spacetimes, it can be slightly modified such that $K$ and $K'$ are compact in space but indefinitely extended in time.}. As a consequence, as one takes the two points further apart, the fastest null curve will go around any localised region rather than going through it. The implication of this theorem is that localised change of the metric cannot induce any time advance asymptotically.

A simple illustration of this theorem can be seen in the Schwarzschild spacetime \eqref{eq:schwarz_D} shown in Figure \ref{fig:gao-wald-theorem}. As we take the two points $p$ and $q$ to infinity, the distance of closest approach of the prompt curve connecting them grows indefinitely with $R$ (see \eqref{eq:r0_Sch_Line}).
\subsection{Extension to Gravitational EFTs} \label{sec:theorem}
With our conclusions on the asymptotic structure of the $FFR$ theory for $D = 4$ in mind, we would like to propose an extension of the Gao-Wald theorem to generic gravitational EFTs in $D=4$ spacetimes. To do so, one must quantify in what sense the light cone approaches that of Schwarzschild asymptotically. We introduce the following definition:
\begin{definition}[Asymptotically Schwarzschild Spacetime]
    A stationary spacetime $(\mathcal{M},\ell^{eff})$, where $\ell^{eff}$ is the broadest effective lightcone, is said to be \emph{asymptotically-Schwarzschild} if there exists a Schwarzschild reference metric $g^{Sch}_{\mu\nu}$ and a compact spatial ball $\Sigma$ of coordinate radius $\sigma$ such that for all curves $\Gamma$ causal with respect to $\ell^{eff}$ outside $\Sigma$:
    \begin{equation}\label{eq:thm_bounds}
        g^{Sch}_{\mu\nu} \deriv{\Gamma^\mu}{t}\deriv{\Gamma^\nu}{t} \le \frac{c}{r^{1+\epsilon}}\,,\quad r>\sigma
    \end{equation}
    for constants $\epsilon>0$ and $c>0$, where $r$ denotes the areal radius coordinate.
    \label{def:asymp_Schwarz}
\end{definition}
We introduce the ball $\Sigma$ to separate the region near the black hole, where large deviations from the Schwarzschild metric $g^{Sch}_{\mu\nu}$ are allowed, and the outer region where only mild superluminality is permitted. We require \eqref{eq:thm_bounds} to quantify how fast the effective light cone should approach the Schwarzschild light cone.

Any computation that does not involve points inside $\Sigma$ can be performed using this fictitious Schwarzschild metric up to a uniformly bounded error, as shown in Appendix \ref{app:bounds}. To make the notation less cumbersome, we will report approximate Schwarzschild results leaving the error bound implicit.
\begin{theorem}
Given a 4-dimensional stationary asymptotically-Schwarzschild spacetime with an effective light cone $\ell^{eff}$, and provided we forbid curves from entering hyperbolicity breaking regions (if they exist), then for every space region $K$ there exists another region $K' \supset K$ such that for all $p,q \notin K'$ the prompt causal curve $\gamma$ connecting $p$ to $q$ does not intersect the region $K$.
    \label{th:theorem}
\end{theorem}
\begin{figure}[h]
    \centering
    \includegraphics[width=\linewidth]{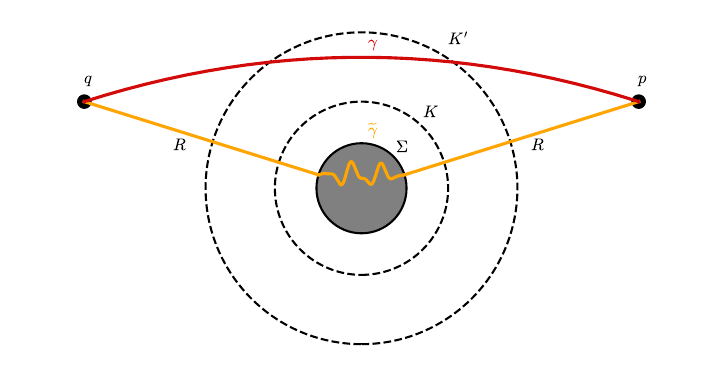}
    \caption{The setup for Theorem \ref{th:theorem}.} 
    \label{fig:theorem_fig}
\end{figure}
Qualitatively, this theorem suggests that the asymptotic structure of a black hole in any modified gravity theory in 4-dimensional spacetimes will have the same asymptotic causal structure as that of General Relativity.
\begin{proof}
We focus on regions $K$ that are balls containing and concentric with $\Sigma$ as shown in Figure \ref{fig:theorem_fig}. We will explicitly exhibit the region $K'[K]$ for these balls, while for any other $K$ that is not a ball centred around the origin we can reduce to the previous case by assigning the region $K'[K]$ to be $K'[\tilde K]$, where $\tilde K$ is a ball containing $K$ and concentric with $\Sigma$.

Recalling the distance of closest approach for the Schwarzschild geodesic \eqref{eq:r0_Sch_Line} in $D=4$, we temporarily take $K'$ to be a concentric ball with radius given by
\begin{equation} \label{eq:thm_rkp1}
    r_s r_{K'} \gtrsim r^2_{K}.
\end{equation}
We will later have to impose an additional condition on $r_{K'}$.
Condition \eqref{eq:thm_rkp1} ensures that for any points $p,q \notin K'$, for simplicity both at the same distance $R$ from the origin, there exists a Schwarzschild null geodesic $\gamma$ connecting them never intersecting $K$. We begin by looking at the technically more challenging case where $p$ and $q$ are antipodal, leaving the other case for later. The distance of closest approach is \eqref{eq:r0_Sch_Line}, and the time of flight of this curve is then given by
\begin{equation}
    \Delta t_\gamma \simeq 2R + r_s \left(\ln\frac{4R}{r_s}+1\right) \,. \label{eq:t_gamma}
\end{equation}

Next we consider a generic null curve $\Tilde{\gamma}$ connecting $p$ and $q$ that enters $K$. We want to show that $\Tilde{\gamma}$ is slower than $\gamma$. If $\Tilde{\gamma}$ does not enter $\Sigma$, then it remains in the Schwarzschild-like region and is slower than $\gamma$ by standard results in general relativity. If $\Tilde{\gamma}$ enters $\Sigma$ at least once, then it will encounter deviations from Schwarzschild and could have a chance at becoming the prompt curve. We denote $A$ and $B$ as the entry and exit points on the surface of $\Sigma$ respectively, which allows us to divide $\Tilde{\gamma}$ into three parts: $\Tilde{\gamma}_{p A}$ and $\Tilde{\gamma}_{B q}$ which lie outside $\Sigma$ and $\Tilde{\gamma}_{AB}$ which lies inside $\Sigma$. We now explore all possible paths and all possible $A,\,B$ and minimise each time interval independently from the others (meaning that the choices of $A,\,B$ need not coincide for the three curves). Then
\begin{equation}
    \min \Delta t_{\Tilde{\gamma}} \ge \min \Delta t_{\Tilde{\gamma}_{pA}} + \min \Delta t_{\Tilde{\gamma}_{AB}} + \min \Delta t_{\Tilde{\gamma}_{Bq}}.
\end{equation}
For $\Tilde{\gamma}_{pA}$ (respectively, $\Tilde{\gamma}_{Bq}$), the path with the minimum time possible between $p\,(q)$ and the sphere $\Sigma$ is a radial geodesic,
\begin{align}
    \min \Delta t_{\Tilde{\gamma}_{pA}} = \min \Delta t_{\Tilde{\gamma}_{Bq}} = \int_{\sigma}^{R} \dd r \left(1 - \frac{r_s}{r}\right)^{-1} \approx (R - \sigma) + r_s \ln \left(\frac{R}{\sigma}\right).
\end{align}
As for $\Tilde{\gamma}_{AB}$, provided that we limit ourselves to Lorentzian metrics by excising any hyperbolicity breaking regions from our spacetime, the curve $\tilde\gamma_{AB}$ can only go forward in time which implies $\min \Delta t_{\Tilde{\gamma}_{AB}} \ge 0$. Finally, we can compare the time of flight between $\gamma$ and $\Tilde{\gamma}$:
\begin{equation}\label{eq:thm_rkp2}
    \Delta t_\gamma - \min \Delta t_{\Tilde{\gamma}} \le 2 \sigma + r_s + r_s \ln \left(\frac{4 \sigma^2}{R r_s}\right)\,.
\end{equation}
Considering that $\sigma$ and $r_s$ are fixed, while $R$ can be taken arbitrarily large, we can find a $r_{K'}=\min R$ that guarantees that the above difference is negative (even after accounting for the bounded error); this guarantees that the Schwarzschild geodesic is faster. This happens when
\begin{equation}
    \label{eq:thm_rkp2b}
    r_{K'} > \frac{4\sigma^2}{r_s} e^{\frac{2\sigma}{r_s}+1}
\end{equation}
Notice that \eqref{eq:thm_rkp2b} is a condition on $r_{K'}$ that does not scale with $r_K$, contrary to \eqref{eq:thm_rkp1}, but the length scale involved can be exponentially large. Putting the two conditions together, we derive the final condition on $r_{K'}$
\begin{equation}
    \label{eq:thm_rkpfinal}
    r_{K'} > \max\bigg\{ \frac{r_K^2}{r_s},\; \frac{4\sigma^2}{r_s} e^{\frac{2\sigma}{r_s}+1} \bigg\}\,.
\end{equation}
For initial and final points at a distance at least \eqref{eq:thm_rkpfinal} from the origin, the prompt geodesic connecting them is the Schwarzschild-like one.

For the more general setup where $p,\,q$ are not antipodal, a similar argument can be applied\footnote{This is also valid in arbitrary dimensions.}. The analogue of \eqref{eq:t_gamma} is 
\begin{equation} \label{eq:thm_nonantipodal}
    \Delta t_\gamma \simeq 2R \cos\tilde\phi + r_s \left( 2\ln\frac{2}{\sin\tilde\phi} +1 \right).
\end{equation}
The time of flight of $\tilde\gamma$ goes as $\sim 2R$, and the time advances inside $\Sigma$ can be bounded in the same way as before. Then the triangle inequality, and not the logarithm particular of $D=4$, ensures that the Schwarzschild-like geodesic will always be faster at $\mathcal{O}(R^1)$.  Equivalently, the time of flight of the Schwarzschild geodesic \eqref{eq:thm_nonantipodal} is shorter because of the extra cosine. This concludes the proof.
\end{proof}

We conclude commenting on some examples where \eqref{eq:thm_bounds} holds. Clearly the $FFR$ effective metric meets this criterion, because beyond GR corrections vanish faster in $1/r$. In addition, to stress that the proof of the theorem relied on stationarity and never on rotational symmetry, we point out that also the Kerr light cone does approach the Schwarzschild light cone in the sense of \eqref{eq:thm_bounds}.
\section{IR Regularisation for \texorpdfstring{$D = 4$}{D = 4}}
In the previous section, we have shown that for 4-dimensional, stationary, asymptotically-Schwarzschild spacetimes, the asymptotic causal structure is identical to Schwarzschild regardless of any beyond-GR corrections. As a consequence, asymptotic causality does not provide a meaningful constraint in this setting. The origin of this can be traced to the logarithmic IR divergence in $D = 4$: the gravitational time delay grows without bound when taking the two points to infinity, overwhelming any finite time advances induced by higher-order operators. 

This motivates the question of what happens when we introduce an IR regulator in the theory. We will explore two ways to proceed: placing the theory of interest in an asymptotically-AdS background in which the AdS radius provides a natural and covariant cut-off, or introducing a hard cut-off by working at a finite distance from the black hole. 
\subsection{Covariant Cut-Off: Asymptotically-AdS Background}\label{AdS}
Gao and Wald established another theorem in \cite{Gao:2000ga},  stating that for generic asymptotically-AdS spacetimes that satisfy the null energy condition, any prompt causal curve between two points on the conformal boundary must lie entirely on the boundary. Thus, the AdS radius provides a natural cut-off for the theory while also preserving a well-defined notion of asymptotic causality.

We are now going to show that, while asymptotic causality bounds are unavailable for 4-dimensional asymptotically-flat spacetimes, this is not the case for 4-dimensional asymptotically-AdS spacetimes: the time delay is finite and allows meaningful comparisons with the AdS background.  However, we also show that this procedure does not admit a flat-space limit: the logarithmic divergence rears its head when the AdS radius is taken to be very large.

We start with the $D = 4$ AdS-Schwarzschild metric with AdS radius $L$ given by
\begin{equation}
    \dd s^2 = -\left(1 - \frac{r_s}{r} + \frac{r^2}{L^2}\right) \dd t^2 + \left(1 - \frac{r_s}{r} + \frac{r^2}{L^2}\right)^{-1} \dd r^2 + r^2 \dd \Omega_2^2. \label{eq:AdS_Schwarz_Metric}
\end{equation}
We consider the same setup as in Section \ref{ssec:null_geodesics} and \ref{ssec:tof}. To determine the null geodesics connecting the two points, we evaluate the bending angle using \eqref{eq:bending_angle_SAdS_Integral}. Due to a cancellation of the $L$-dependent terms in the integral, the expression for the bending angle is exactly identical\footnote{This is due to the fact that we parametrised the integral in terms of the distance of closest approach $r_0$ rather than the impact parameter $b$. In the latter case, the results will be the same as the asymptotically-flat case only to leading order in $L$.} to that of the asymptotically-flat case. As a result, when sending the two points along a straight line, the distance of closest approach of the null geodesic will grow indefinitely with $R$ as \eqref{eq:r0_Sch_Line}, which for $D=4$ reduces to
\begin{equation}
    r_0 = \sqrt{R r_s}. \label{eq:r0_Sch_D=4}
\end{equation}
For geodesics, we can evaluate the time of flight using \eqref{eq:t_geo_SAdS_Integral}. Expanding to $\mathcal{O}(r_s)$, and working in the region where $R \gg r_0$ we find that the geodesic time of flight is given by 
\begin{equation}
    \Delta t_{geo} = \pi L + 2r_s \coth^{-1} \left(\sqrt{1 + \frac{r_0^2}{L^2}}\right).
\end{equation}
The first term is the pure AdS contribution, while the second term is the gravitational time delay induced by the black hole. Since the distance of closest approach of the geodesic grows with $R$ according to \eqref{eq:r0_Sch_D=4}, in this limit the gravitational time delay vanishes and one recovers
\begin{equation}
    \lim_{R\rightarrow\infty} \Delta t_{geo} = \pi L
\end{equation}
showing that asymptotic AdS–Schwarzschild geodesics are completely insensitive to the presence of the black hole. This is similar to the behaviour found earlier in $D>4$ asymptotically-flat Schwarzschild spacetime.

In AdS backgrounds, null curves which probe the bulk experience a finite time delay with respect to the boundary. This delay can, in principle, compete with time advances induced by higher-order operators and allow for null curves through the bulk which are faster than the AdS-Schwarzschild geodesic at infinity. %
This makes asymptotically-AdS backgrounds a viable setting for deriving causality constraints on gravitational EFTs. However, this comes at a cost of an explicit dependence on the AdS regulator. If we try to take the flat space limit by taking $L \gg r_0$, the curve going through the bulk will be delayed relative to the AdS background by
\begin{equation}
    \Delta t_{geo}^{{\text{beyond GR}}} - \pi L \approx 2r_s \ln \left(\frac{2L}{r_0}\right) {- \Delta t^{adv}}
\end{equation}
{and the log} will dominate any finite time advance coming from beyond-GR effects in the strict limit $L\to\infty$. We therefore conclude that AdS causality bounds do not admit a smooth flat-space limit: the IR effects that obstruct asymptotic causality in flat space reappear as $L$ is taken to infinity even after the AdS regulator has rendered intermediate steps finite. 
\subsection{Hard Cut-Off: Working at Finite Distance}\label{finitedistance}
Another natural question is whether causality constraints can be obtained remaining at finite distance. In other words, can we cut off the IR logarithm at a finite $R$ so that it does not dominate against the EFT corrections? At this point readers might object that statements about superluminality at finite distance is prone to coordinate ambiguities \cite{Gao:2000ga}. 
 In $D>4$, there exists a natural class of coordinate systems such that the corrections to Minkowski decay as $1/r^{(D-3)}$, within which it is unambiguous to discuss superluminality and to compare with the asymptotic Minkowski metric \cite{Camanho:2014apa}. In $D = 4$, corrections to Minkowski go as $r_s/r$, provided one works in so-called `good' coordinates, and this fixes the coefficient of the logarithmic term in the time delay. Finite pieces, however, are coordinate dependent: even within the class of `good' coordinates, different choices change the time delay by finite pieces of order $r_s$\footnote{This can be seen, for example, by comparing the Shapiro time delay computed in standard Schwarzschild coordinates with that in isotropic coordinates.}. Consequently, in the regime where {$R$ is large enough so that} the logarithm is large compared to unity, the constant pieces are negligible and the result is robust; in particular, the GR time delay can be unambiguously compared with the time advance induced by EFT operators.

Consider \eqref{eq:Schvsstraight}, now without taking the limit $R \to \infty$ but rather taking $R$ finite and exponentially large such that the logarithm wins against the finite $\mathcal{O}(r_s)$ piece and we compare this time delay with the time advance induced by the $FFR$ term. Should a theory that admits a time-advance in this sense be considered pathological?

Typically, time advances are considered pathological because they are associated with the possibility of generating closed time-like curves (CTCs). The basic logic is that once a time advance exists, one can boost the solution so that in the new frame the signal propagates backwards in coordinate time. Completing the CTC then requires a second background boosted in the opposite direction. In order to superimpose the two boosted black holes, they must be displaced in the direction orthogonal to their motion. We now argue that this is in tension with the previous requirement of taking $R$ exponentially large, and hence that superluminality in the sense previously defined need not be connected to the generation of CTCs.
\begin{figure}
    \centering
    \includegraphics[width=0.7\linewidth]{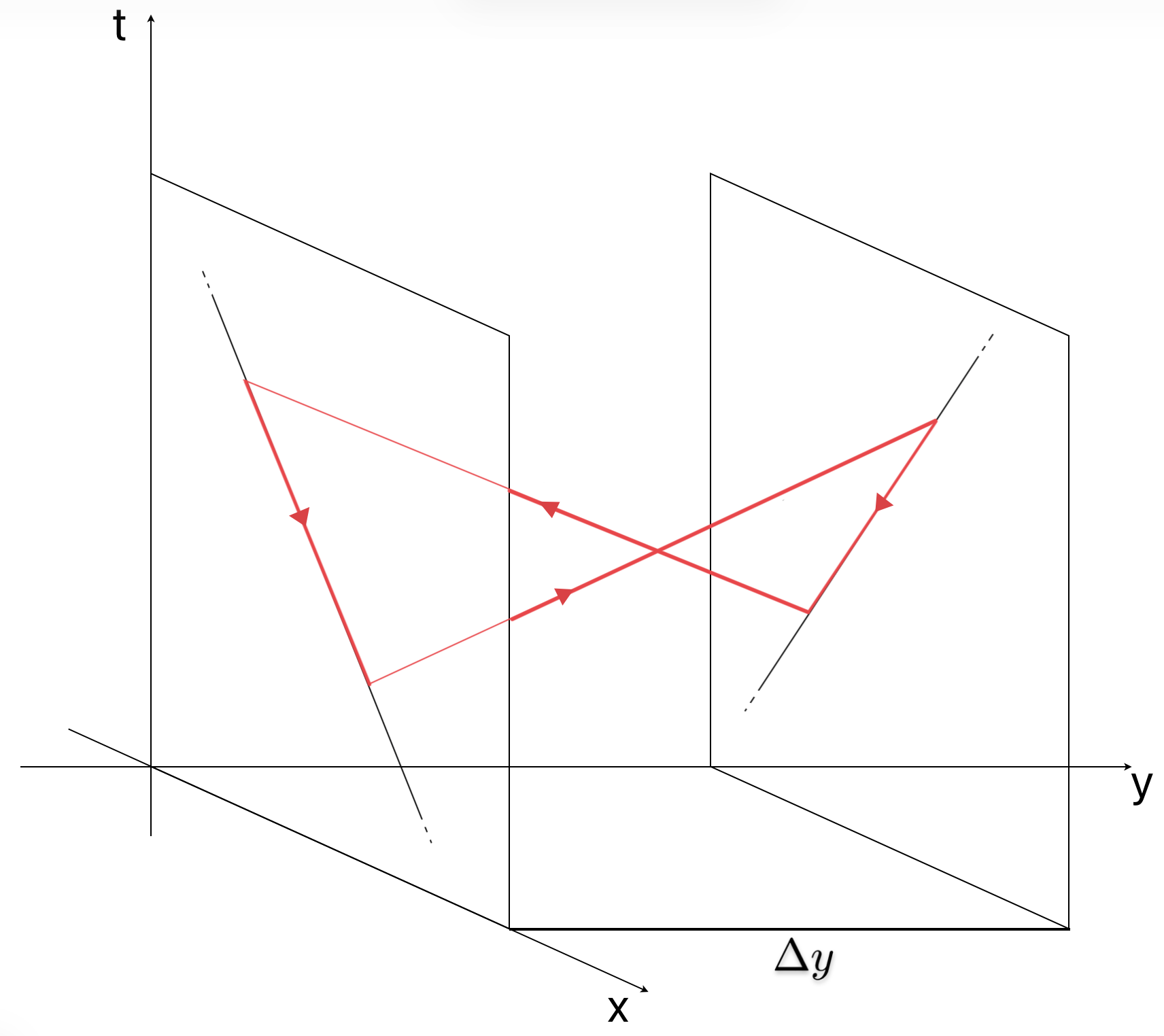}
    \caption{Superposition of two boosted solutions to construct a CTC.}
    \label{fig:ctc}
\end{figure}
In the black hole's rest frame, one must propagate an exponentially large distance $R \sim r_s e^N$, with $N \gg 1$, to accumulate a time advance of order $r_s$ (since we are working with exponential distance scales, we assume $r_s \sim \sqrt{\alpha}$). We now boost this solution so that in the new coordinates, the signal travels backwards in coordinate time: 
\begin{equation}\label{eq:SR}
    \Delta t' = \gamma(\Delta t - \beta \Delta x) \;.
\end{equation}
To get $\Delta t' <0$, a very large boost is required with $\beta \gtrsim 1 - e^{-N}$. This boost enhances the corrections to the metric in the orthogonal direction: the solution becomes similar to a shock-wave. 
At linear order, the boosted Schwarzschild solution reads
\begin{equation}
   h_{\mu\nu} = \frac{2GM}{r} (\eta_{\mu\nu} + 2 u_\mu u_\nu) \;.
\end{equation}
Here $u_\mu = (-\gamma, 0, 0, \gamma \beta)$ is the 4-velocity of the black hole and $r$ is the distance written in the new coordinates (we drop the prime on the coordinates for brevity): $r = \sqrt{x^2 +y^2 +\gamma^2 (z- \beta t)^2}$.
To superimpose two solutions, the transverse distance must be large enough that the metric is in the linear regime, $|h_{\mu\nu}| \ll 1$. This requires $\Delta y \gg r_s \gamma^2 \simeq r_s e^N$. To complete the CTC, the signal has to move in the transverse direction from one boosted black hole to the other, which requires a traversal time $\Delta t_\perp\gtrsim r_s e^{N}$. However, from \eqref{eq:SR}, the time advance from the boost grows only as $\gamma \simeq e^{N/2}$ which is parametrically smaller than $e^N$ and therefore insufficient to compensate. It does not seem possible to build a CTC, at least with a naive superposition of boosted solutions. 

In conclusion, in $D=4$, there is no obvious route from superluminality at large but finite distance to the possibility of building a CTC. Notice that in higher dimensions there is no need to consider exponentially large distances: it is enough to travel for few $r_s$ to see superluminality. Therefore there is no need to consider exponentially large boost factors $\gamma$ and the obstruction above disappears. This difference can also be seen at the level of the shock-wave solutions considered in \cite{Camanho:2014apa}. In $D=4$ the shock wave is completely delocalised in the orthogonal direction. Correspondingly, the construction of CTCs (their Appendix G) only works for $D>4$. Of course, it is still possible that CTCs may form in the non-linear regime where the two black holes cannot be simply superimposed.

\section{Concluding Remarks}

The main result of this work is a sharp dichotomy between four and higher spacetime dimensions regarding the {\it asymptotic} causal structure in gravitational EFTs.

In $D>4$, higher-derivative operators can generate time advances that survive in the asymptotic limit, leading to genuine violations of the causal structure determined by the background metric. These violations constrain the EFT cut-off to lie parametrically below the perturbative strong-coupling scale, $\Lambda \lesssim 1/\sqrt{\alpha}$.

In $D=4$, the difference in time of flight between null curves that 
get close to a black hole and the Schwarzschild-like geodesic diverges logarithmically with the distance of the endpoints to the source. This divergence ensures that, for any stationary asymptotically-Schwarzschild spacetime, prompt null curves between distant points never probe the near-horizon region where higher-derivative corrections are localised (Theorem~\ref{th:theorem}). The asymptotic causal structure is therefore universally that of Schwarzschild, regardless of the specific form of the irrelevant operators.

One might hope to evade this conclusion by remaining at finite, though exponentially large, distances, where the logarithm is large but not infinite. We have argued in Section~\ref{finitedistance} that the resulting time advances do not lead to closed timelike curves in any obvious way: the exponentially large boosts required to reverse the time ordering also expand the gravitational field in the transverse direction, preventing the naive superposition of two boosted black holes.

{A distinct, local form of causality constraint does, however, survive in $D = 4$. As discussed in Section~\ref{phenomenology}, the effective metric of the $FFR$ theory develops a hyperbolicity-breaking region at the scale $r_*$, below which it changes signature and the effective light cone ceases to be Lorentzian. Requiring the EFT to remain well-posed even close to the black hole therefore imposes a cut-off bound $\Lambda \lesssim 1/r_* \ll (M_P/\alpha)^{1/3}$, ruling out the strong-coupling scale as the UV cut-off. This constraint is local rather than asymptotic, and is therefore insensitive to the logarithmic IR divergence that obstructs asymptotic causality in $D = 4$. It is weaker than the asymptotic bounds available in $D > 4$ — where prompt null curves directly probe the near-horizon region — but it is free from coordinate ambiguities and does not require taking signals to infinite distance. Although we established this argument for the specific example of the $FFR$ operator, the mechanism is generic: whenever higher-derivative corrections drive the effective metric of some propagating mode to lose hyperbolicity at a scale $r_{\rm_{hyp}}$ on a black-hole background, a cut-off bound $\Lambda \lesssim 1/r_{\rm_{hyp}}$ follows. The bound applies regardless of spacetime dimension, relying only on the requirement that the EFT admit a well-posed initial-value problem.}

{So far, we have discussed causality via propagation on classical backgrounds. In the pioneering work of Adams \textit{et al.}~\cite{Adams:2006sv}, causality in EFTs was investigated through two complementary frameworks. The first is the classical-propagation approach taken in this work; the second is based on analyticity and positivity of the S-matrix. In the gravitational case, this second approach encounters the same infrared logarithm.}

The $t$-channel graviton pole obstructs the standard $t\to 0$ limit; this is resolved by writing smeared dispersion relations~\cite{Caron-Huot:2021rmr}. However, a logarithm from the graviton loop remains, requiring an IR regulator---for example, a cosmological background---to extract finite bounds on Wilson coefficients \footnote{Refs.~\cite{Beadle:2024hqg,Beadle:2025cdx,Chang:2025cxc} extend the methodology of~\cite{Caron-Huot:2021rmr} by consistently incorporating gravitational loops}. 
This logarithm in $D=4$ is the scattering-amplitude manifestation of the same IR effect we identify geometrically. The absence of IR divergences in $D>4$ means that standard positivity bounds derived from dispersion relations apply without obstruction, consistent with our time-delay results.

The S-matrix approach has the practical advantage of not requiring a verification that closed timelike curves can be constructed: a bound follows once the dispersive representation is established. For this reason, it has been used extensively, for example, in \cite{Caron-Huot:2022ugt, Henriksson:2022oeu, Hong:2023zgm, Albert:2024yap,Xu:2024iao, Dong:2024omo, Dong:2025dpy}.
On the other hand, it requires assumptions about the UV behaviour of the amplitude, in particular Regge boundedness, and the analyticity properties themselves become less transparent once Lorentz symmetry is broken by a cosmological regulator.

A recent proposal~\cite{Bellazzini:2025bay} offers a different strategy: replacing ordinary amplitudes with stripped amplitudes $\mathcal{M}_{\mathcal{E}}$, labelled by a detector energy resolution $\mathcal{E}$. Bounds derived from $\mathcal{M}_{\mathcal{E}}$ retain a residual $\log\mathcal{E}$ dependence but avoid the forward-limit divergence that plagues the gravitational S-matrix. Whether the positivity of $\mathcal{M}_{\mathcal{E}}$ has a direct spacetime interpretation---i.e., whether it encodes a well-defined condition on signal propagation---is an open question left for future work. Another open direction, complementary to the AdS analysis of Section~\ref{AdS}, is to extend the geometric analysis of this work to asymptotically de Sitter backgrounds\footnote{For a discussion on causality conditions in de Sitter spacetime see also \cite{Bittermann:2022hhy,McLoughlin:2025shj}.}, where the cosmological horizon provides a physically motivated IR regulator in $D=4$. 

\acknowledgments
We would like to thank D.~McLoughlin, M.~Mirbabayi, A.~Nicolis and F.~Serra for helpful discussions. 
We appreciate valuable discussions and feedback from participants of the `Constraining Effective Field Theories Without Lorentz' workshop in Trieste and `Theoretical Tools for Gravitational Wave Physics' workshop in Zurich. 
WPMcB would like to extend his gratitude to SNS for their hospitality while this work was in progress. 
BB is supported by the U.S. Department of Energy under grant number DESC0019470 and by the Heising-Simons Foundation `Observational Signatures of Quantum Gravity' collaboration grant 2024-5305. ET is partially supported by the Italian MIUR under contract 20223ANFHR (PRIN2022).

\appendix
\section{Derivation of the \texorpdfstring{$FFR$}{FFR} Effective Metrics} \label{app:FFR_effectivemetric}
In this appendix, we will derive the effective metrics for photons\footnote{We have omitted the derivation of the effective metrics of gravitons for brevity, but it can be shown that they follow the light cones of the background metric and thus do not dictate the causal structure of the theory.} described by the action \eqref{eq:L_FFR}, following closely the methods outlined in \cite{Davies:2021frz} and \cite{Reall:2021voz}.

The action that we are interested is \eqref{eq:L_FFR}, while the background metric is that of Schwarzschild \eqref{eq:schwarz_D}. The equation of motion for photons is given by
\begin{equation}
    E^{\mu} = \frac{\delta S}{\delta A_\mu} = \nabla_\nu F^{\mu\nu} - 4\alpha R^{\mu\nu}_{~~~\rho\sigma}\nabla_{\nu}F^{\rho\sigma} = 0.
\end{equation}
If we linearise the equation around the background solution, we obtain
\begin{equation}\label{characteristiceq}
    \frac{\partial E^\mu}{\partial (\partial_\alpha \partial_\beta A_\nu)} \partial_\alpha \partial_\beta \delta A_\nu + \dots = (g^{\mu \nu} g^{\alpha\beta} - g^{\mu \alpha}g^{\nu\beta} - 8\alpha R^{\mu\alpha\nu\beta})\partial_\alpha \partial_\beta \delta A_\nu + \dots = 0
\end{equation}
where the ellipses denote terms with fewer than two derivatives acting on $\delta A_\nu$. The operator in the front of $\partial_\alpha \partial_\beta \delta A_\nu$ determines how the linear perturbations propagate in this system. For a non-trivial propagation, this operator must be non-invertible. This motivates the definition of the \emph{principal symbol}
\begin{equation}
    P(x, \xi)^{\mu\nu} =  \frac{\partial E^\mu}{\partial (\partial_\alpha \partial_\beta A_\nu)} \xi_\alpha \xi_\beta = g^{\mu\nu} \xi^2 - \xi^{\mu}\xi^{\nu} - 8\alpha R^{\mu\alpha\nu\beta}\xi_{\alpha}\xi_\beta \label{eq:principal_symbol}
\end{equation}
with $\xi_\mu$ being a general covector. We can regard \eqref{eq:principal_symbol} as a matrix acting on gauge equivalence classes of polarisation vectors $t_\mu$. To see this, we note that
\begin{equation}
    P(x, \xi)^{\mu\nu} (t_\mu + \beta \xi_\mu) = P(x, \xi)^{\mu\nu} t_\mu.
\end{equation}
These classes correspond to physical polarisations. Thus, we say that $\xi_\mu$ is characteristic if there exists a non-zero equivalence class of $t_{\mu}$ satisfying the characteristic equation
\begin{equation}
     P(x, \xi)^{\mu\nu} t_{\mu} = 0.
\end{equation}
To solve this, we decompose the polarisation vector into a transverse component $\hat{t}_\mu$, such that $\hat{t}_\mu \xi^\mu = 0$, and a longitudinal one. The characteristic equation then becomes
\begin{align}
    P(x, \xi)^{\mu\nu} t_{\mu} &= P(x, \xi)^{\mu\nu} \hat{t}_{\mu}\\
    &= (g^{\mu\nu} \xi^2 - \xi^{\mu}\xi^{\nu} - 8\alpha R^{\mu\alpha\nu\beta}\xi_\alpha \xi_\beta)\hat{t}_{\mu}\\
    &= (g^{\mu\nu} \xi^2 - 8\alpha R^{\mu\alpha\nu\beta} \xi_\alpha \xi_\beta)\hat{t}_{\mu} = 0
\end{align}
which can be rearranged as
\begin{equation}
    8\alpha R_\mu^{~~\alpha\nu\beta} \xi_\alpha \xi_\beta \hat{t}_{\nu} = \xi^2 \hat{t}_\nu. \label{eq:eigen_char}
\end{equation}
The characteristics can be finally determined by solving the eigenvalue problem \eqref{eq:eigen_char}. The two $\xi^2 \neq 0$ solutions, in $D=4$ are:
\begin{equation}
    \xi^2=\frac{12 \alpha  r_S \left(\xi_2^2\sin^2{\theta}+ \xi_3^2\right)}{r^5\sin^2{\theta}\left(1+4 \frac{\alpha r_S}{r^3}\right)}\>,\>\xi^2=-\frac{12 \alpha  r_S \left(\xi_2^2\sin^2{\theta}+ \xi_3^2\right)}{r^5\sin^2{\theta}\left(1-8 \frac{\alpha r_S}{r^3}\right)}
\end{equation}
Since $\xi_{\mu}$ is a covector, the above norms correspond to the two effective metrics
\begin{equation}
    g^{\mu \nu}={\rm diag}(-1/f(r),r(r),c_V(r)\Omega_{AB}/r^2)\>,\>g^{\mu \nu}={\rm diag}(-1/f(r),r(r),c_S(r)\Omega_{AB}/r^2)
\end{equation}
respectively. The first eigenvalue's eigenvector corresponds to the vector (parity odd) polarisation, and it points purely in the $\theta$ direction if the photons' momentum is confined on the equatorial plane. Conversely, the second eigenvalue corresponds to the scalar (parity even) polarisation. 
This twofold splitting persists in $D > 4$ dimensions because the $D - 2$ transverse polarisations always decompose into two distinct eigenspaces: modes purely tangent to the $S^{D-2}$ sphere, and mode containing a radial projection. Due to the background symmetries, the eigenvalue problem \eqref{eq:eigen_char} remains highly degenerate, consistently yielding exactly two effective metrics, regardless of the number of spacetime dimensions.

\section{Null Curves in Spherically Symmetric Spacetimes} \label{app:general_metric}
A general $D$-dimensional static and spherically symmetric spacetime can be cast in the form:
\begin{equation}
    \dd s^2 = -f(r) \dd t^2 + f(r)^{-1} \dd r^2 + \frac{r^2}{c(r)} \dd \Omega_{D - 2}^2, \label{eq:gen_sphere_metric}
\end{equation}
with $\dd \Omega_{D-2}^2$ being the metric of $S^{D-2}$. A null curve is a trajectory such that $\dd s^2 = 0$.
\subsection{Null Geodesics}
Introducing the polar coordinates $(\theta_1, \theta_2,...,\theta_{D-3}, \phi)$ on $S^{D-2}$. The spherical symmetry of this spacetime allows us to confine all geodesics on the equatorial plane i.e. $\theta_1 = \theta_2 = ... = \theta_{D-3} = \pi/2$. Associated to two Killing vectors, $\partial/\partial t$ and $\partial/\partial \phi$, there exist two conserved quantities along geodesics:
\begin{align}
    e &= f(r)\dot{t} = 1, \label{eq:consv_e}\\
    \ell &= \frac{r^2}{c(r)} \dot{\phi} \label{eq:consv_l} \;.
\end{align}
We denote with a dot the derivative with respect to the affine parameter, and for null geodesics we can normalise the affine parameter such that $e = 1$. The conserved quantities allow us to rewrite $\dd s^2 = 0$ as the radial equation 
\begin{equation}
    \frac{1}{2} \dot{r}^2 + V_{\text{eff}} (r) = \frac{1}{2}
\end{equation}
in which
\begin{equation}
    V_{\text{eff}}(r) = \frac{\ell^2}{2r^2} f(r)c(r) \label{eq:v_eff}
\end{equation}
is the effective potential experienced by the photon moving along the null geodesic. The bending angle and the time of flight along the null geodesic connecting $p$ and $q$ are then given by the following integrals:
\begin{equation} \label{eq:app_deltaphiT_integrals}
\begin{split}
    \Delta \phi &= 2\int_{r_0}^{R} \dd r \frac{c(r)}{r} \left\{\frac{r^2}{r_0^2} f(r_0)c(r_0) - f(r)c(r)\right\}^{-\frac{1}{2}},\\
    \Delta t_{geo} &= 2\int_{r_0}^{R} \dd r \frac{1}{f(r)} \left\{1 - \frac{r_0^2}{r^2} \frac{f(r)c(r)}{f(r_0)c(r_0)}\right\}^{-\frac{1}{2}}
\end{split}
\end{equation}
in which $r_0$ is the \emph{distance of closest approach} of the null geodesic. The relationship between $r_0$ and the angular momentum $\ell$ is obtained setting $\dot r =0$ at $r =r_0$:
\begin{equation}
    \ell = \frac{r_0}{\sqrt{f(r_0)c(r_0)}}\,. \label{eq:l_r0}
\end{equation}
Because of our choice of affine parameter such that the energy is normalized to $e=1$, the physical interpretation of $\ell$ is that it is also the impact parameter $b$ of the trajectory.

\subsection{General Null Curves}
We might want to compare the time of flight of the null geodesic to that of a general null curve connecting the same two points. We will limit ourselves to curves on the equatorial plane, although the idea can be extended to a general curve. Consider a curve described by
\begin{equation}
    h(r,\phi) = 0.
\end{equation}
Along this curve, we have
\begin{equation}
    \dd h = \frac{\partial h}{\partial r} \dd r + \frac{\partial h}{\partial \phi} \dd \phi = 0.
\end{equation}
Solving for $\dd \phi$ and plugging into $\dd s^2 = 0$ we obtain
\begin{equation}
    \Delta t = \int \dd r \frac{1}{f(r)} \left\{1 + r^2 \frac{f(r)}{c(r)} \left[\frac{\partial h/\partial r}{\partial h/\partial \phi}\right]^2\right\}^{\frac{1}{2}}.
\end{equation}
An example is a straight line described by
\begin{equation}
    r \sin \phi = d.
\end{equation}
Following the previous steps, the time of flight along this straight line is
\begin{equation}
    \Delta t_{str} = 2 \int_{r_0}^{R} \dd r \frac{1}{f(r)} \left\{1 + \frac{f(r)}{c(r)} \frac{d^2}{r^2 - d^2}\right\}^{\frac{1}{2}}.
\end{equation}
\subsection{Perturbative Calculations: \texorpdfstring{$FFR$}{FFR} Theory} \label{app:perturb_calcs_FFR}
In this section, the calculations of each quantity is done with respect to the largest effective metric given by \eqref{eq:eff_metric_C}. We will restrict to the case $\alpha > 0$, but we note that the perturbative results—at least to first order—are independent of the sign of $\alpha$. For the Schwarzschild case, set $\alpha = 0$.
\subsubsection{Bending Angle}
For simplicity, we change the integration variable to $z = r_0/r$ such that $0 < z \leq 1$. The integral is given by
\begin{multline}
    \Delta \phi = 2 \int_{\frac{r_0}{R}}^{1} \dd z \frac{1}{z} \left[\frac{1 + \frac{c_1 \alpha r_s^{D-3}}{r_0^{D-1}} z^{D-1}}{1 - \frac{c_2 \alpha r_s^{D-3}}{r_0^{D-1}} z^{D-1}}\right] \cdot\\
    \cdot\left\{\frac{1}{z^2}\left[1 - \left(\frac{r_s}{r_0}\right)^{D-3}\right]\left[\frac{1 + \frac{c_1 \alpha r_s^{D-3}}{r_0^{D-1}}}{1 - \frac{c_2 \alpha r_s^{D-3}}{r_0^{D-1}}}\right] - \left[1 - \left(\frac{r_s}{r_0}z\right)^{D-3}\right]\left[\frac{1 + \frac{c_1 \alpha r_s^{D-3}}{r_0^{D-1}}z^{D-1}}{1 - \frac{c_2 \alpha r_s^{D-3}}{r_0^{D-1}}z^{D-1}}\right]\right\}^{-\frac{1}{2}}.
\end{multline}
Expand to $\mathcal{O}\left[\left(\frac{r_s}{r_0}\right)^{D-3},\left(\frac{\alpha r_s^{D-3}}{r_0^{D-1}}\right)\right]$, neglecting higher-order terms,
\begin{equation}
    \Delta \phi \approx 2 \int_{\frac{r_0}{R}}^{1} \dd z \left\{ \frac{1}{(1 - z^2)^\frac{1}{2}} + \frac{1}{2} \left(\frac{r_s}{r_0}\right)^{D-3} \frac{1 - z^{D-1}}{(1-z^2)^\frac{3}{2}} - \frac{1}{2}(c_1 + c_2) \frac{\alpha r_s^{D-3}}{r_0^{D-1}} \frac{1 - 2z^{D-1} + z^{D+1}}{(1 - z^2)^\frac{3}{2}}\right\}
\end{equation}
and integrate to obtain
\begin{equation}
    \Delta \phi \approx \pi - 2 \sin^{-1} \left(\frac{r_0}{R}\right) + \left(\frac{r_s}{r_0}\right)^{D-3} \sqrt{\pi} (D - 1) \left[1 - \frac{\alpha}{r_0^2} (c_1 + c_2)\frac{D-2}{D-1}\right] \frac{\Gamma(\frac{D}{2})}{2\Gamma(\frac{D+1}{2})}  \;.\label{eq:bending_angle_full}
\end{equation}
In the last expression we also expanded assuming $R \gg r_0$ and $R \gg r_s$.
\subsubsection{Geodesic Time of Flight}
The integral is given by
\begin{multline}
    \Delta t_{geo} = 2r_0 \int_\frac{r_0}{R}^{1} \dd z \frac{1}{z^2} \left[1 - \left(\frac{r_s}{r_0}z\right)^{D-3}\right]^{-1} \\
    \left\{1 - z^2 \frac{\left[1 - \left(\frac{r_s}{r_0}z\right)^{D-3}\right]}{\left[1 - \left(\frac{r_s}{r_0}\right)^{D-3}\right]}\left[\frac{1 + \frac{c_1 \alpha r_s^{D-3}}{r_0^{D-1}}z^{D-1}}{1 - \frac{c_2 \alpha r_s^{D-3}}{r_0^{D-1}} z^{D-1}}\right]\left[\frac{1 + \frac{c_1 \alpha r_s^{D-3}}{r_0^{D-1}}}{1 - \frac{c_2 \alpha r_s^{D-3}}{r_0^{D-1}}}\right]^{-1}\right\}^{-\frac{1}{2}}.
\end{multline}
Expand to $\mathcal{O}\left[\left(\frac{r_s}{r_0}\right)^{D-3},\left(\frac{\alpha r_s^{D-3}}{r_0^{D-1}}\right)\right]$, neglecting higher-order terms,
\begin{equation}
    \Delta t_{geo} \approx 2r_0 \int_{\frac{r_0}{R}}^1 \dd z \left\{\frac{1}{z^2(1 - z^2)^\frac{1}{2}} + \frac{1}{2}\left(\frac{r_s}{r_0}\right)^{D-3} \frac{2z^{D-5} - 3z^{D-3} + 1}{(1-z^2)^{3/2}} - \frac{1}{2}(c_1 + c_2)\frac{\alpha r_s^{D-3}}{r_0^{D-1}} \frac{1-z^{D-1}}{(1-z^2)^{\frac{3}{2}}}\right\}\label{eq:t_geo_approx_integral}
\end{equation}
and integrate to obtain
\begin{equation}
    \Delta t_{geo} = 2\sqrt{R^2 - r_0^2} + r_0 \left(\frac{r_s}{r_0}\right)^{D-3} \sqrt{\pi} \frac{(D-1)(D-3)}{D-4}\left[1 - \frac{\alpha}{r_0^2} (c_1 + c_2) \frac{D-4}{D-3}\right] \frac{\Gamma\left(\frac{D}{2}\right)}{2\Gamma\left(\frac{D+1}{2}\right)} \label{eq:D_t_geo}
\end{equation}
This is valid for $D>4$. As for $D = 4$, the result diverges due to the log-divergence of the Schwarzschild time delay as we take $R\rightarrow \infty$. We can handle the $D=4$ case separately by plugging $D=4$ before integrating \eqref{eq:t_geo_approx_integral}
\begin{equation}
    \Delta t_{geo} \approx 2\sqrt{R^2 - r_0^2} + r_s\left[2  \ln \left(\frac{2R}{r_0}\right) + 1\right] - \frac{24 \alpha r_s}{r_0^2}. \label{eq:4d_t_geo}
\end{equation}
\subsubsection{Straight Line Time of Flight} 
For simplicity, we change the integration variable to $z = d/r$ such that $0 < z \leq 1$. The integral is then given by
\begin{multline}
    \Delta t_{str} = 2d \int_{\frac{d}{R}}^{1} \dd z \frac{1}{z^2} \left[1 - \left(\frac{r_s}{d}z\right)^{D-3}\right]^{-1} \\
    \left\{1 + \frac{z^2}{1 - z^2} \left[1 - \left(\frac{r_s}{d}z\right)^{D-3}\right] \left[\frac{1 - \frac{c_2 \alpha r_s^{D-3}}{d^{D-1}} z^{D-1}}{1 + \frac{c_1 \alpha r_s^{D-3}}{d^{D-1}}z^{D-1}}\right]\right\}^{\frac{1}{2}}.
\end{multline}
Expand to $\mathcal{O}\left[\left(\frac{r_s}{d}\right)^{D-3},\left(\frac{\alpha r_s^{D-3}}{d^{D-1}}\right)\right]$ (neglecting higher-order terms),
\begin{equation}
    \Delta t_{str} \approx 2d \int_{\frac{d}{R}}^{1} \dd z \left\{\frac{1}{z^2 (1 - z^2)^{\frac{1}{2}}} + \left(\frac{r_s}{d}\right)^{D-3} \frac{z^{D-5} - \frac{1}{2}z^{D-3}}{(1 - z^2)^{\frac{1}{2}}} -\frac{1}{2}(c_1 + c_2) \frac{\alpha r_s^{D-3}}{d^{D-1}} \frac{z^{D-1}}{(1 - z^2)^{\frac{1}{2}}}\right\}, \label{eq:t_str_approx_integral}
\end{equation}
and integrate to obtain
\begin{equation}
    \Delta t_{str} \approx 2\sqrt{R^2 - d^2} + d \left(\frac{r_s}{d}\right)^{D-3} \sqrt{\pi} \frac{D-1}{D-4} \left[1 - \frac{\alpha}{d^2} (c_1 + c_2)\frac{D-4}{D-1}\right] \frac{\Gamma\left(\frac{D}{2}\right)}{2\Gamma\left(\frac{D+1}{2}\right)}. \label{eq:D_t_str}
\end{equation}
This is valid for $D>4$. As for $D = 4$, the result diverges due to the log-divergence of the Schwarzschild time delay as we take $R\rightarrow \infty$. We can handle the $D=4$ case separately by plugging $D=4$ before integrating \eqref{eq:t_str_approx_integral}
\begin{equation}
    \Delta t_{str} \approx 2\sqrt{R^2 - d^2} + r_s\left[2  \ln \left(\frac{2R}{d}\right) - 1\right] - \frac{8 \alpha r_s}{d^2}. \label{eq:4d_t_str}
\end{equation}

It is worth mentioning here that the results \eqref{eq:bending_angle_full} for $D=4$ and \eqref{eq:4d_t_str} for the $FFR$ results matches with S-matrix calculations (see, e.g., (4.74) and (4.75) of \cite{AccettulliHuber:2020oou}).
\subsection{Perturbative Calculations: AdS-Schwarzschild}
In this section, the calculations of each quantity is done with respect to the $D=4$ AdS-Schwarzschild metric given by \eqref{eq:AdS_Schwarz_Metric}. 
\subsubsection{Bending Angle}
The integral is given by
\begin{align}
    \Delta \phi &= 2\int_{r_0}^{R} \dd r \frac{1}{r} \left\{\frac{r^2}{r_0^2} \left(1 - \frac{r_s}{r_0} + \frac{r_0^2}{L^2}\right) - \left(1 - \frac{r_s}{r} + \frac{r^2}{L^2}\right)\right\}^{-\frac{1}{2}}\\
    &= 2\int_{r_0}^{R} \dd r \frac{1}{r} \left\{\frac{r^2}{r_0^2} \left(1 - \frac{r_s}{r_0}\right) - \left(1 - \frac{r_s}{r}\right)\right\}^{-\frac{1}{2}} \label{eq:bending_angle_SAdS_Integral}
\end{align}
which is identical to the asymptotically flat Schwarzschild case, meaning that the previous results (with $D = 4$) still apply i.e.
\begin{equation}
    \Delta \phi \approx \pi - 2 \sin^{-1} \left(\frac{r_0}{R}\right) + 2\frac{r_s}{r_0}. \label{eq:bending_angle_AdsSchwarz}
\end{equation}
\subsubsection{Geodesic Time of Flight}
The integral is given by
\begin{equation}
    \Delta t_{geo} = 2\int_{r_0}^{R} \dd r \left(1 - \frac{r_s}{r} + \frac{r^2}{L^2}\right)^{-1} \left\{1 - \frac{r_0^2}{r^2} \left(1 - \frac{r_s}{r} + \frac{r^2}{L^2}\right)\left(1 - \frac{r_s}{r_0} + \frac{r_0^2}{L^2}\right)^{-1}\right\}^{-\frac{1}{2}}. \label{eq:t_geo_SAdS_Integral}
\end{equation}
Expanding to $\mathcal{O}\left(r_s\right)$ (neglecting higher-order terms),
\begin{multline}
    \Delta t_{geo} =  2\int_{r_0}^{R} \dd r \frac{1}{(1 + \frac{r^2}{L^2})} \left\{1 - \frac{r_0^2}{r^2} \frac{1 + \frac{r^2}{L^2}}{1 + \frac{r_0^2}{L^2}}\right\}^{-\frac{1}{2}} + \\
    \left(\frac{1}{r\left(1+\dfrac{r^{2}}{L^{2}}\right)^{2}
  \sqrt{\dfrac{L^{2}\left(r^{2}-r_{0}^{2}\right)}{r^{2}\left(L^{2}+r_{0}^{2}\right)}}
}
+
\frac{
  r r_{0}
  \sqrt{\dfrac{L^{2}\left(r^{2}-r_{0}^{2}\right)}{r^{2}\left(L^{2}+r_{0}^{2}\right)}}
  \left(L^{2}+r^{2}+r r_{0}+r_{0}^{2}\right)
}{
  2 L^{2}\left(1+\dfrac{r^{2}}{L^{2}}\right)
  (r-r_{0})(r+r_{0})^{2}
}
\right) r_{s},
\end{multline}
and integrate to obtain
\begin{equation}
    \Delta t_{geo} = 2L \tan^{-1} \left(\frac{\sqrt{R^2 - r_0^2}}{\sqrt{L^2 + r_0^2}}\right) + 2r_s \coth^{-1} \left(\sqrt{1 + \frac{r_0^2}{L^2}}\right). \label{eq:delta_t_geo_AdSSch}
\end{equation}

\section{Error Bounds on Time of Flight}
\label{app:bounds}
In our proof in Section \ref{sec:theorem}, we argued that outside the spherical region $K$ we can treat the effective metric as Schwarzschild up to a minor correction \eqref{eq:thm_bounds}. Astute readers might be concerned that when taking $p$ and $q$ to infinity, the minor corrections might add up along the path and introduce errors that could grow without bound, spoiling the proof altogether. 

In this appendix we will show that the error in the time of flight is finite and bounded for any relevant null curve, when computing it using the Schwarzschild metric $g_{\mu\nu}^{Sch}$ instead of using the correct effective light cone $\ell^{eff}_{\mu\nu}$ (which often comes from an effective metric $g^{eff}_{\mu\nu}$, although we will not assume this). This is true even accounting for the fact that geodesic curves of the effetive metric do not precisely coincide with those of the Schwarzschild metric. Qualitatively, this means that the computation done in Section \ref{sec:theorem} correctly captures the logarithmic term crucial to the proof, while any other terms that might be introduced by $g^{eff}_{\mu\nu}$ will be finite and sub-leading in the limit as $p$ and $q$ approach infinity. This proof also applies to $D>4$, although in this case the logarithmic term is absent and finite corrections become relevant.

\begin{theorem}

    Let $\Gamma$ be a prompt causal curve of the effective light cone $\ell^{eff}$ obeying \eqref{eq:thm_bounds}, lying entirely outside the ball $\Sigma$ and connecting two space points $p$ and $q$. Let $\Gamma^s$ be the null geodesic of the Schwarzschild metric connecting $p$ and $q$.
    We are interested in bounding the error we make when approximating the time of flight $\Delta t_\Gamma$ of $\Gamma$ with the time of flight $\Delta t_{\Gamma^s}$ of $\Gamma^s$, the latter being more easily calculable. We will prove
    \begin{equation} \label{eqappC:goal}
        |\Delta t_{\Gamma} - \Delta t_{\Gamma^s}| \leq C
    \end{equation}
    where $C$ is a constant.
\end{theorem}

\begin{proof}

    It is convenient to relate $\Gamma$ and $\Gamma^s$ by defining a third curve $\hat\Gamma$, null in the Schwarzschild light cone and tracing the path of $\Gamma$ in space
    \begin{equation}
        \hat{\Gamma}^i(x) \equiv \Gamma^i(x), \quad g^{Sch}_{\mu\nu} \frac{\dd \hat{\Gamma}^\mu}{\dd t}\frac{\dd \hat{\Gamma}^\nu}{\dd t} \equiv 0\,. \label{eq:def_gamma_Sch}
    \end{equation}
    We will prove \eqref{eqappC:goal} by putting together three observations. The first is that $\hat\Gamma$, while causal, is not generically a Schwarzschild geodesic contrary to $\Gamma^s$, so
    \begin{equation}
        \label{eqappC:fact1}
        \Delta t_{\Gamma^s}\le \Delta t_{\hat\Gamma}
    \end{equation}
    because geodesics are prompt in the Schwarzschild metric.

    The second observation is that, according to \eqref{eq:thm_bounds}, $\ell^{eff}$ allows small superluminalities compared to the the Schwarzschild light cone. As a consequence the set of trajectories causal in $\ell^{eff}$ enlarges that of Schwarzschild, so from the assumption that $\Gamma$ is prompt we get
    \begin{equation}
    \label{eqappC:fact2}
        \Delta t_{\Gamma}\le \Delta t_{\Gamma^s}\,.
    \end{equation}

    The third inequality to conclude the argument will be of the form $\Delta t_{\Gamma^s}-C\le \Delta t_\Gamma$ for some constant $C>0$, independent of the curve $\Gamma$ and its endpoints.
    
    To proceed we define $u_\Gamma$ as
    \begin{equation}
        u_{\Gamma} = g_{\mu\nu}^{Sch} \frac{\dd \Gamma^\mu}{\dd t}\frac{\dd \Gamma^\nu}{\dd t}, \quad 0 < u_\Gamma < \frac{c}{r^{1 + \epsilon}}. \label{eq:u_Gamma}
    \end{equation}
    where the bound on $u_\Gamma$ comes from \eqref{eq:thm_bounds}.
    We can solve for $\dd t_{\Gamma}$ to obtain\footnote{We are writing everything in terms of differentials to shorten the notation. The more mathematically inclined can rewrite it as an integral over a certain dummy parameter if one wishes to do so.}
    \begin{equation}
        \dd t_\Gamma = \sqrt{\frac{g_{ij}^{Sch} \dd \Gamma^i \dd \Gamma^j}{-g_{tt}^{Sch} + u_\Gamma}}. \label{eq:dt_Gamma}
    \end{equation}
    Given that $-g_{tt}^{Sch} > 0$ and $u_{\Gamma} > 0$, the denominator of \eqref{eq:dt_Gamma} can be bounded using
    \begin{equation}
        \sqrt{\frac{1}{-g^{Sch}_{tt}}} - \sqrt{\frac{1}{-g^{Sch}_{tt} + u_\Gamma}} \leq \frac{u_\Gamma}{2 (-g_{tt}^{Sch})^{\frac{3}{2}}}
    \end{equation}
    which gives
    \begin{equation}
        \dd t_\Gamma \geq \sqrt{\frac{g_{ij}^{Sch} \dd \Gamma^i \dd \Gamma^j}{-g_{tt}^{Sch}}} - \frac{u_\Gamma}{2} \sqrt{\frac{g_{ij}^{Sch} \dd \Gamma^i \dd \Gamma^j}{(-g_{tt}^{Sch})^3}}. \label{eq:dt_Gamma2}
    \end{equation}
    The first term in the right hand side of \eqref{eq:dt_Gamma2} is the differential time of flight of $\hat{\Gamma}$ as defined in \eqref{eq:def_gamma_Sch}, while the second term is the deviation that should be bounded. Thus, we can rewrite \eqref{eq:dt_Gamma2} as
    \begin{equation}
        \dd t_{\Gamma} - \dd t_{\hat{\Gamma}} \geq - \frac{u_\Gamma}{2} \sqrt{\frac{g_{ij}^{Sch} \dd \Gamma^i \dd \Gamma^j}{(-g_{tt}^{Sch})^3}}. \label{eq:dtG-dtGhat}
    \end{equation}
    Next, we consider the right hand side of \eqref{eq:dtG-dtGhat}. The error term can be written using \eqref{eq:gen_sphere_metric} as
    \begin{align}
        - \frac{u_\Gamma}{2} \sqrt{\frac{g_{ij}^{Sch} \dd \Gamma^i \dd \Gamma^j}{(-g_{tt}^{Sch})^3}} &= - \frac{u_\Gamma}{2} \frac{1}{[f(r)]^2}\sqrt{\dd r^2 + r^2 f(r) \dd \Omega_2^2} \, .
    \end{align}
    We will now bound this expression from below. Using the fact that for Schwarzschild
    \begin{equation}
        1 - \frac{r_s}{\sigma} \leq f(r) \leq 1\, , \quad
        r>\sigma\, ,
    \end{equation}
    along with \eqref{eq:u_Gamma}, the inequality \eqref{eq:dtG-dtGhat} becomes
    \begin{equation}
         \dd t_{\Gamma} - \dd t_{\hat{\Gamma}} \geq - \frac{1}{2\left(1 - \frac{r_s}{\sigma}\right)^2} \frac{c}{r^{1 + \epsilon}} \sqrt{\dd r^2 + r^2 (\dd \theta^2 + \sin^2 \theta \dd \phi^2)}.
    \end{equation}
    where we recognise the Euclidean space metric $\dd s^2$ under the square root. Integrating both sides from $p$ to $q$ we derive
    \begin{equation}
        \label{eqappC:dtG-dtGhat_integrated}
        \Delta t_\Gamma - \Delta t_{\hat\Gamma} \ge - \frac{c}{2\left(1-\frac{r_s}{\sigma}\right)^2}\int \frac{1}{r^{1+\epsilon}}\dd s
    \end{equation}
    The last integral can be bounded under mild technical assumptions on the curve $\Gamma$, which we expect to be fulfilled for prompt curves. In particular we expect $\Gamma$ to be divided into two arcs, one approaching the origin and receding from it, both obeying
    \begin{equation}
    \label{eqappC:regularity_condition}
        r(s) \ge k (1+s)^p\, , \quad k>0, \, p> \frac{1}{1+\epsilon}\, .
    \end{equation}
    This assumption qualitatively requires that $\Gamma$ always possesses a non-vanishing radial component of motion. The integral in \eqref{eqappC:dtG-dtGhat_integrated} can now be bounded by putting the end points of the curve at infinity, giving
    \begin{equation}
        \int \frac{\dd s}{r^{1+\epsilon}} \le \frac{2}{k^{1+\epsilon}}\int_0^\infty \frac{\dd s}{(1+s)^{p(1+\epsilon)}} \le \frac{2}{k^{1+\epsilon} (p+p\epsilon-1)}\, .
    \end{equation}
    This result, along with \eqref{eqappC:dtG-dtGhat_integrated} and \eqref{eqappC:fact1}, shows that
    \begin{equation}
        \Delta t_\Gamma \ge \Delta t_{\hat\Gamma}-C \ge t_{\Gamma^s} -C\, , \quad
        C = \frac{c}{\left(1-\frac{r_s}{\sigma}\right)^2k^{1+\epsilon} (p+p\epsilon-1)}\, .
    \end{equation}
    Combining this inequality with \eqref{eqappC:fact2} we finally conclude that
    \begin{equation}
    \Delta t_{\Gamma^s}-C \le \Delta t_\Gamma \le \Delta t_{\Gamma^s}
    \Rightarrow |\Delta t_\Gamma-\Delta t_{\Gamma^s}|\le C\, .
\end{equation}
    
\end{proof}

\bibliographystyle{JHEP} %
\bibliography{draft_bib}

\end{document}